\documentclass[fleqn,reqno]{article}
\usepackage{amsmath,amssymb,amsthm,mathrsfs}

\usepackage[margin=1.0in]{geometry}
\usepackage{enumerate}

\usepackage{graphics,geometry,epsfig}

\theoremstyle{plain}

\newtheorem{theorem}{Theorem}
\newtheorem{proposition}[theorem]{Proposition}
\newtheorem{lemma}[theorem]{Lemma}
\newtheorem{corollary}[theorem]{Corollary}
\newtheorem{definition}[theorem]{Definition}

\numberwithin{theorem}{section} \numberwithin{equation}{section}

\newcommand{\donothing}[1]{}

\newcommand{\DETAILS}[1]{}

\begin{document}


\newcommand{\nc}{\newcommand}

\nc{\be}{\begin{equation}}
\nc{\la}{\label}
\nc{\ba}{\begin{array}}
\nc{\ea}{\end{array}}
\nc{\bs}{\begin{split}}
\nc{\es}{\end{split}}

\nc{\J}{\mathbb J}

\newcommand{\R}{\mathbb{R}}
\newcommand{\C}{\mathbb{C}}
\newcommand{\Z}{\mathbb{Z}}
\newcommand{\T}{\mathbb{T}}

\newcommand{\cA}{\mathcal{A}}
\newcommand{\cB}{\mathcal{B}}
\newcommand{\cC}{\mathcal{C}}
\newcommand{\cD}{\mathcal{D}}
\newcommand{\cE}{\mathcal{E}}
\newcommand{\cF}{\mathcal{F}}
\newcommand{\cG}{\mathcal{G}}
\newcommand{\cH}{\mathcal{H}}
\newcommand{\cI}{\mathcal{I}}
\newcommand{\cJ}{\mathcal{J}}
\newcommand{\cK}{\mathcal{K}}
\newcommand{\cL}{\mathcal{L}}         
\newcommand{\cM}{\mathcal{M}}         
\newcommand{\cN}{\mathcal{N}}         
\newcommand{\cO}{\mathcal{O}}         
\newcommand{\cP}{\mathcal{P}}         
\newcommand{\cQ}{\mathcal{Q}}
\newcommand{\cR}{\mathcal{R}}
\newcommand{\cS}{\mathcal{S}}
\newcommand{\cT}{\mathcal{T}}
\newcommand{\cU}{\mathcal{U}}
\newcommand{\cV}{\mathcal{V}}
\newcommand{\cW}{\mathcal{W}}
\newcommand{\cX}{\mathcal{X}}
\newcommand{\cY}{\mathcal{Y}}
\newcommand{\cZ}{\mathcal{Z}}
\newcommand{\Lat}{\mathcal{L}}

\newcommand{\Ge}{\mathcal{G}}
\newcommand{\Per}{\mathcal{P}}

\nc{\E}{{\cal E}}

\nc{\e}{\epsilon}

\nc{\G}{\Gamma}
\nc{\g}{\gamma}
\nc{\al}{\alpha}
\nc{\del}{\delta}

\nc{\Lam}{\Lambda}

\nc{\lam}{\lambda}
\nc{\lamb}{\lambda}
\nc{\taul}{\lambda}
\nc{\om}{\omega}
\nc{\Om}{\Omega}
\nc{\Omt}{\tilde{\Omega}}
\nc{\ta}{\tau}
\nc{\w}{\omega}
\nc{\io}{\iota}
\nc{\h}{\theta}
\nc{\z}{\zeta}
\nc{\s}{\sigma}
\nc{\Si}{\Sigma}
\nc{\vphi}{\varphi}

\nc{\bP}{\bar{P}}
\nc{\bQ}{\bar{Q}}
\nc{\bL}{\bar{L}}

\nc{\chit}{\tilde{\chi}}
\nc{\It}{\tilde{I}}
\nc{\alt}{\tilde{\alpha}}
\nc{\Mt}{\tilde{M}}
\nc{\gt}{\tilde{\gamma}}
\nc{\Jt}{\tilde{J}}

\renewcommand{\L}{\mathcal{L}}

\nc{\oP}{\overline P}

\newcommand{\LAT}{\mathcal{L}}

\nc{\zb}{\underbar{z}}
\nc{\pb}{\underbar{p}}
\nc{\bb}{\underbar{b}}
\nc{\zbt}{\underbar{z}(t)}

\nc{\dA}{\nabla_A}

\newcommand{\p}{\partial}
\newcommand{\pt}{\partial_t}
\newcommand{\ptt}{\partial_t^2}
\newcommand{\n}{{\nabla}}

\newcommand{\Curl}{\operatorname{curl}}
\newcommand{\curl}{\operatorname{curl}}
\newcommand{\CURL}{\operatorname{curl}}
\newcommand{\divv}{\operatorname{div}}
\newcommand{\Div}{\operatorname{div}}
\newcommand{\DIV}{\operatorname{div}}
\renewcommand{\div}{\operatorname{div}}
\newcommand{\grad}{\operatorname{grad}}

\newcommand{\Cov}[1]{\nabla_{\!\!#1}}
\newcommand{\COVGRAD}[1]{\nabla_{\!\!#1}}

\newcommand{\COVLAP}[1]{\Delta_{\!#1}}

\renewcommand{\Re}{\operatorname{Re}}
\renewcommand{\Im}{\operatorname{Im}}
\newcommand{\im}{\operatorname{Im}}
\newcommand{\re}{\operatorname{Re}}

\newcommand{\one}{\mathbf{1}}
\newcommand{\id}{{\bfone}}
\nc{\bfone}{{\bf 1}}

\nc{\ran}{\rangle}
\nc{\lan}{\langle}

\newcommand{\ra}{\rightarrow}


\newcommand{\NULL}{\operatorname{Null}}
\newcommand{\RANGE}{\operatorname{Ran}}
\newcommand{\Null}{\operatorname{Null}}
\newcommand{\Ran}{\operatorname{Ran}}
\newcommand{\Tr}{\operatorname{Tr}}
\newcommand{\diag}{\operatorname{diag}}
\newcommand{\dist}{\operatorname{dist}}

\newcommand{\Lpsi}[2]{\mathscr{L}_{#2}^{}(\tau)}
\newcommand{\Hpsi}[2]{\mathscr{H}_{#2}^{}(\tau)}
\newcommand{\LA}[2]{\vec{\mathscr{L}}_{}^{}(\tau)}
\newcommand{\HA}[2]{\vec{\mathscr{H}}_{}^{}(\tau)}

\newcommand{\ls}{\lesssim}

\newcommand{\CELLAVG}[2]{\left\langle #2 \right\rangle_{#1}}
\newcommand{\DOT}[3]{\left\langle #2, #3 \right\rangle_{#1}}
\newcommand{\NORM}[2]{\left\| #2 \right\|_{#1}}

\newcommand{\TwoByOne}[2]{\left( \begin{array}{c} #1 \\ #2 \end{array} \right)}
\newcommand{\TwoByTwo}[4]{\left( \begin{array}{cc} #1 & #2 \\ #3 & #4 \end{array} \right)}
\newcommand{\ThreeByThree}[9]{\left( \begin{array}{c} #1 & #2  & #3  \\ #4 & #5  & #6  \\ #7 & #8  & #9  \end{array} \right)}
\newcommand{\FourByOne}[4]{\left( \begin{array}{c} #1 \\ #2 \\ #3 \\ #4 \end{array} \right)}

\nc{\nj}{(n_j)}
\nc{\nk}{(n_k)}
\nc{\nl}{(n_l)}
\nc{\nr}{(n_r)}
\nc{\nt}{(n_t)}

\nc{\vz}{v_{\zb,\chi}}
\nc{\vzo}{v_{\zb_0,\chi_0}}
\nc{\vzt}{v_{\zb(t),\chi(t)}}
\nc{\psiz}{\psi_{\zb,\chi}}
\nc{\Az}{A_{\zb,\chi}}
\nc{\Bz}{B_{\zb,\chi}}
\nc{\jz}{j_{\zb,\chi}}
\nc{\Tz}{T^{(\zb,\chi)}}
\nc{\Gz}{G^{(\zb,\chi)}}
\nc{\Lz}{L_{\zb,\chi}}
\nc{\Lzo}{L_{\zb_0,\chi_0}}
\nc{\Pz}{P_{\zb,\chi}}
\nc{\bPz}{\bar{P}_{\zb,\chi}}
\nc{\phiz}{\phi_{\zb,\chi,\pb,\zeta}}
\nc{\phizo}{\phi_{\zb_0,\chi_0,\pb_0,\zeta_0}}
\nc{\phizt}{\phi_{\zb(t),\chi(t),\pb(t),\zeta(t)}}
\nc{\wz}{w_{\zb,\chi,\pb,\zeta}}
\nc{\wzo}{w_{\zb_0,\chi_0,\pb_0,\zeta_0}}

%

\title{Magnetic Vortices, 
 Abrikosov Lattices and Automorphic Functions
\thanks{\copyright 2013 by the author. This paper may be reproduced, in its entirety, for non-commercial purposes.}}
\author{
I. M. Sigal\footnote{Department of Mathematics, University of Toronto, Toronto, ON, Canada, M5S 2E4}}
\date{August 2013}

\maketitle

\begin{abstract}
We address the macroscopic theory of superconductivity - the 
Ginzburg-Landau theory. This theory 
is based on the celebrated  Ginzburg - Landau equations. First developed to explain and predict properties of superconductors, these equations form an integral part - Abelean-Higgs component - of the standard model of particle physics and, in general, have a profound influence on physics well beyond their original designation area. 

We present recent results 
and review earlier works  
 involving key solutions of these equations - the magnetic vortices (of Nielsen-Olesen ( Nambu) strings in 
 particle physics) and vortex lattices, their existence, stability and dynamics, and how they relate to the modified theta functions appearing in number theory. Some automorphic functions appear naturally and play a key role in this theory.\\
 


 \bigskip
\noindent This paper is contribution to the International Conference on Applied Mathematics, Modeling and Computational Science (AMMCS-2013),
 Waterloo, Ontario, Canada from August, 2013.\\ 
 
\end{abstract}

\tableofcontents

\section{Introduction}
In this contribution we present some recent results on the Ginzburg-Landau equations of superconductivity and  review appropriate background.
The Ginzburg-Landau equations describe the key mesoscopic and macroscopic properties of
superconductors and form the basis of the phenomenological theory of
superconductivity. They are thought of as the result of
coarse-graining the Bardeen-Cooper-Schrieffer microscopic theory, and were derived
from the latter 
by Gorkov \cite{Gorkov}. (Recently, the rigorous derivation in the case of non-dynamic magnetic fields was achieved in \cite{FHSS}.) 

These equations appear also in particle physics, as the Abelean-Higgs model, which is the simplest, and arguably most
important, ingredient of the standard
model \cite{Witten}.
 Geometrically, they are the simplest equations describing the interaction of the electro-magnetic
field and a complex field, and can be thought of as 
the `Dirichlet'
problem for a connection of $U(1)$-principal bundle and a section of associated vector bundle.

One of the most interesting mathematical and physical
phenomena connected with Ginzburg-Landau equations
is the presence of {\it vortices} in their solutions.
Roughly speaking, a vortex is a spatially localized structure in the solution, characterized by a non-trivial topological degree (winding number). 
It represents a localized defect where the normal state
intrudes, and magnetic flux penetrates. it is called the magnetic vortex.

Vortices exist on their own, 
 or, as 
predicted by A. Abrikosov  \cite{Abr} in 1957, 
they  can be arrayed in a lattice pattern. (In 2003, Abrikosov received the Nobel Prize for this discovery.)

Individual vortices and vortex lattices is the subject of this contribution.  In it we present already classical results on 
the former and recent results on the latter. It can be considered as an update of the review \cite{GST}, from which for convenience of the reader, we reproduce some material. 

Like the latter review, we  do not 
discuss the 
two important areas,  the $\kappa\ra \infty$ regime 
(the quasiclassical limit of the theory) and the linear eigenvalue problem related
to the second critical magnetic field. Fairly extensive reviews of these problems are given in \cite{SS} and \cite{Helffer}, respectively. We also mention the book \cite{bbh} which inspired much of the activity in this area.

\DETAILS{In the last decade or so, vortex solutions have become the object of
intense mathematical study, in several directions.
One direction is to consider the singular limit
(``extreme type II'') $\kappa \to \infty$ (on a bounded
domain), in which vortices become point defects whose locations
are determined by some reduced finite-dimensional problem --
see, for example the books of Bethuel-Brezis-H\'elein~\cite{bbh}
(for a model problem without magnetic field) and
Serfaty-Sandier~\cite{SS}.

The Ginzburg-Landau equations became an object of mathematical
investigation soon after their invention in the early 1950's,
and in the last 15 years this subject flourished, partly due to the
influence of the book \cite{bbh} on the related Ginzburg-Landau
equations of superfluidity. Two aspects of this development were
recently extensively reviewed: the regime of large Ginzburg-Landau
parameter $\kappa$, the quasiclassical limit of the theory,
in \cite{ss}; and the linear eigenvalue problem related
to the second critical magnetic field, in \cite{Helffer}.

In this article, we review a third aspect of the theory:
the key -- {\it vortex} -- solutions, and their dynamics, for any
value of the Ginzburg-Landau parameter $\kappa$. Moreover, we consider
only bulk superconductors filling all $\R^3$, with no variation along
one direction, so that the problem is reduced to one on $\R^2$.
The list of open problems presented in the last section shows  that this area
sits poised for broad and rapid development.}

\medskip

\noindent \textbf{Acknowledgements} \\
The author is are grateful to Stephen Gustafson, Yuri Ovchinnikov and Tim Tzaneteas for many fruitful discussions and collaboration. Author's research is supported in part by NSERC under Grant NA 7901. 


\section{The Ginzburg-Landau Equations}

The Ginzburg-Landau theory (\cite{GL})   gives a macroscopic
description of superconducting materials in terms of a  pair  $(\Psi , A) : \R^d \to \C \times \R^d$, $d=1, 2, 3$, a complex-valued function $ \Psi(x)$, called an {\it order parameter}, so that $|\Psi(x)|^2$ gives the local density
of (Cooper pairs of) superconducting electrons, and the vector field $A(x)$, 
so that $  B(x) := \curl A(x)$ is the magnetic field. In equilibrium, 
they satisfy the system of nonlinear PDE
called the {\it Ginzburg-Landau equations}:
\begin{equation}
\label{GLE}
   \ba{c}
   -\Delta_A \Psi = \kappa^2 (1-|\Psi|^2) \Psi   \\
   \curl^2 A = \Im ( \bar{\Psi} \nabla_A \Psi )
   \ea
\end{equation}
where $\nabla_A = \nabla - iA$, and $\Delta_A = \nabla^2_A$,
the covariant derivative and covariant Laplacian,
respectively, and $\kappa >0$ is a parameter, called the  Ginzburg-Landau parameter,  depends on the material properties of
the superconductor.  For $d=2$, $\curl A := \p_1 A_2 - \p_2 A_1$ is a scalar, and
for scalar $B(x) \in \R$, $\curl B = (\p_2 B, -\p_1 B)$ is a vector.
The vector quantity $  J(x) := \Im ( \bar{\Psi} \nabla_A \Psi )$ is the superconducting current.  (See eg. \cite{Tink,Tilley}). 

\medskip

\noindent
\textit{Particle physics}.
In the Abelian-Higgs model, $\psi$ and $A$ are the Higgs and
$U(1)$ gauge (electro-magnetic) fields, respectively.  Geometrically,
one can think of $A$ as a connection on the principal $U(1)$-bundle
$\R^d  \times U(1),\ d=2, 3$.

\medskip

\noindent
\textit{Cylindrical geometry}.
In the commonly considered idealized situation of a superconductor occupying all space
and homogeneous in one direction, one is led to a problem on $\R^2$
and so may consider $\Psi : \R^2 \to \C$ and  $A : \R^2 \to \R^2$. 
This is the case we deal with in this contribution. 

\subsection{Ginzburg-Landau energy}

The Ginzburg-Landau equations~\eqref{GLE}
are the Euler-Lagrange equations for critical points of the
\textit{Ginzburg-Landau energy functional} (written here for a domain $Q\in \R^2$)
\be \la{GLEn}
  \E_{Q}(\Psi,A) := \frac{1}{2} \int_{Q} \left\{  |\nabla_A \Psi|^2 + (\curl A)^2
  + \frac{\kappa^2}{2}(|\Psi|^2-1)^2 \right\}.
\end{equation}

\medskip

\noindent
{\it Superconductivity:}
In the case of superconductors, the
functional $\E(\psi,A)$ gives the difference in (Helmholtz) free energy (per unit length in the third direction)
between the superconducting and normal states, near the
transition temperature. 

\DETAILS{, and in the absence of an external
magnetic field. In the presence of an external field
$H(x) \hat{k}$ ($\hat{k}$ a unit vector in the direction
in which the superconductor is homogeneous), the term
$(\curl A)^2$ in the energy density is replaced
by $(curl A - H)^2$. If the field strength $H$ is constant
(a common scenario in the literature), this modification
leaves the Euler-Lagrange equations -- the Ginzburg-Landau
equations~\eqref{eq:GL} -- unchanged.}

This energy depends on the temperature (through $\kappa$) and  the average magnetic field, $b=\lim_{Q\ra \R^2}\frac{1}{|Q|}\int_{Q} \curl A$, in the sample, as thermodynamic parameters. Alternatively, one can consider the free energy depending on the temperature and  an applied magnetic field, $h$. For a sample occupying a finite domain $Q$, this leads (through the Legendre transform) to 
the Ginzburg-Landau Gibbs free energy $G_Q(\Psi,A) :=\E_Q(\Psi,A)  -\Phi_Q h,$
where $\Phi_Q=b|Q|=\int_Q \curl A$ is the total magnetic flux through the sample.  
\be \la{GLEn-Gibbs}
  G_{Q}(\Psi,A) := \frac{1}{2} \int_{Q} \left\{  |\nabla_A \Psi|^2 
  + \frac{\kappa^2}{2}(|\Psi|^2-1)^2+ (\curl A-h)^2 \right\}.
\end{equation}
The parameters $b$ or $h$ do not enter the equations \eqref{GLE} explicitly, but they determine the density of vortices, which we describe below.

In what follows we write $\E(\Psi,A)=\E_{\R^2}(\Psi,A)$ and $G(\Psi,A)=G_{\R^2}(\Psi,A)$.

\medskip

\noindent
{\it Particle physics:}
In the particle physics case, the functional $\E(\Psi,A)$ gives
the energy of a static configuration in the $U(1)$
Yang-Mills-Higgs classical gauge theory.


\subsection{Symmetries of the equations}

The Ginzburg-Landau equations~\eqref{GLE} admit several
symmetries, that is, transformations which map solutions to
solutions.

{\it Gauge symmetry}:  for any sufficiently regular function $\gamma : \R^2 \to \R$,
\begin{equation}\label{gauge-sym}
    T^{\rm gauge}_\gamma : (\Psi(x), A(x)) \mapsto (e^{i\gamma(x)}\Psi(x), A(x) + \nabla\gamma(x));
\end{equation}

{\it Translation symmetry}: for any $h \in \R^2$,
\begin{equation}\label{transl-sym}
    T^{\rm trans}_h : (\Psi(x), A(x)) \mapsto (\Psi(x+h), A(x+h));
\end{equation}

{\it Rotation symmetry}: for any $\rho \in SO(2)$, 
\begin{equation}\label{rot-sym}
    T^{\rm rot}_\rho : (\Psi(x), A(x)) \mapsto (\Psi(\rho^{-1} x), \rho^{-1}A((\rho^{-1})^T x)),
\end{equation}

One of the analytically interesting aspects of the
Ginzburg-Landau theory is the fact that, because of the gauge
transformations, the symmetry group is infinite-dimensional.

\subsection{Quantization of flux}

Finite energy states $(\Psi,A)$ are classified by their
topological degree (the winding number of $\psi$ at infinity):
\[
  \deg(\Psi) := {\rm degree}
  \left( \left. \frac{\Psi}{|\Psi|} \right|_{|x|=R}
  : \mathbb{S}^1 \to \mathbb{S}^1 \right),
\]
for $R \gg 1$, since $|\Psi(x)| \to 1$ as $|x| \to \infty$.
For each such state we have the quantization of magnetic flux:
\[
  \int_{\R^2} B(x) dx = 2 \pi \deg (\Psi) \in 2 \pi \Z,
\]
which follows from integration by parts (Stokes theorem)
and the requirement that $|\Psi(x)| \to 1$
and $|\nabla_A \Psi(x)| \to 0$ as $|x| \to \infty$.

For 
vortex lattices (see below) the energy is infinite, but the flux quantization still holds for each lattice cell because of gauge-periodic boundary conditions (see below for details).


\subsection{Homogeneous solutions}

The simplest solutions to the Ginzburg-Landau
equations~\eqref{GLE} are the trivial ones corresponding
to physically homogeneous states:
\begin{enumerate}
\item the perfect superconductor solution, $(\Psi_s, A_s)$, where
$\Psi_s \equiv 1$ and $A_s \equiv 0$
(so the magnetic field $ \equiv 0$),
\item
the normal metal solution, $(\Psi_n, A_n)$,  where $\Psi_n \equiv  0$ and $A_n$ corresponds to a constant magnetic field.
\end{enumerate}
(Of course, any gauge transformation of one of
these solutions has the same properties.)

We see that the perfect superconductor is a solution only when the magnetic  field $B(x)$ is zero. On the other hand, there is a normal solution 
for any  constant magnetic field (to be thought of as determined by applied external magnetic field).


\subsection{Length scales; type I and II superconductors}

Solving the Ginzburg-Landau equations near a flat interface
between normal and superconducting states shows that (in our units),
the magnetic field varies on the length scale $1$, the
{\em penetration depth}, while the order parameter varies on the
length scale $\frac{1}{m_{\kappa}}$, the {\em coherence length},
where $m_\kappa := \min(\sqrt{2}\kappa, 2)$.

The two length length scales $1/m_{\kappa}$ and $1$ coincide at $\kappa=1/\sqrt{2}$. Considering a flat interface between the normal and superconducting states, one can show easily that at this point the surface tension changes sign from positive for $\kappa < 1/\sqrt{2}$ to negative for $\kappa > 1/\sqrt{2}$.

This critical value $\kappa=1/\sqrt{2}$ separates superconductors into
two classes with different properties:

$\kappa < 1/\sqrt{2}$: Type I superconductors,  exhibit first-order (discontinuous, finite size nucleation) phase transitions from the
non-superconducting state to the superconducting state (essentially, all pure metals);

 $\kappa > 1/\sqrt{2}$: Type II superconductors, exhibit
second-order (continuous) phase transitions and the formation of vortex lattices (dirty metals and alloys).

Thus $m_\kappa  <1$ for Type I superconductors and  $m_\kappa > 1$ for Type II superconductors.


\subsection{The self-dual case $\kappa = 1/\sqrt{2}$}

In the {\it self-dual} case $\kappa = 1/\sqrt{2}$ of~\eqref{GLE},
vortices effectively become non-interacting, and there is a rich
multi-vortex solution family. Bogomolnyi~\cite{Bog} found
the topological energy lower bound
\begin{equation}
\label{bog}
  \E(\Psi,A)|_{\kappa = 1/\sqrt{2}} \geq \pi \; | \deg(\Psi) |
\end{equation}
and showed that this bound is saturated (and hence the
Ginzburg-Landau equations are solved) when certain
{\it first-order} equations are satisfied.


\subsection{Critical magnetic fields}
In superconductivity there are several critical magnetic fields,  two of which (the first and the second  critical magnetic fields) are of special importance:

$h_{c1}$ is the field at which the first vortex enters the superconducting sample.
\smallskip

$h_{c2}$ is the field at which a mixed state bifurcates from the normal one. 
\smallskip

\noindent
(The critical field $h_{c1}$ is defined by the condition $G_{}(\Psi_s, A_s)=G_{}((\Psi^{(1)}, A^{(1)})$, where   $(\Psi_s, A_s)$ is the perfect superconductor solution, defined above, and  $(\Psi^{(1)}, A^{(1)})$ is the 1-vortex solution, defined below, while $h_{c2}$, by the condition that the linearization of the l.h.s. of \eqref{GLE} on the normal state $(\Psi_n, A_n)$ has zero eigenvalue. One can show that $h_{c2}=\kappa^2$.) 

 For type I superconductors $h_{c1} > h_{c2}$ and for type II superconductors $h_{c1} < h_{c2}$. In the former case, the vortex states have relatively large energies, i.e. are metastable, and therefore are of little importance.

For type II superconductors, there are two important regimes to consider: 1) average magnetic fields per unit area, $b$, are less than but sufficiently close to  $h_{c2}$,
\be \la{regime1} 0< h_{c2}-b \ll h_{c2}
\end{equation} 
and 2) the external (applied) constant magnetic fields, $h$,  are  greater than but sufficiently close to  $h_{c1}$,
\be \la{regime2} 
 0< h-h_{c1} \ll h_{c1}.
 \end{equation}
The reason the first condition involves $b$, while the second $h$ is that the first condition comes from the Ginzburg-Landau equations (which do not involve $h$), while the second from the Ginzburg-Landau Gibbs free energy. 

%
\DETAILS{For an external field $h_a>h_{c1}$ and with the intervortex interaction neglected, we have $G/|Q|=m(E^{(1)}-\Phi^{(1)} H_a)$, where $m$ 
is the number of vortices per unit area. Thus the behaviour of $G$, as a function of $m$, changes at $H_a=H_{c1}
:=\frac{E^{(1)}}{\Phi^{(1)}}$. If $H_a<H_{c1}$, then the Gibbs energy is lowered by having no vortices, and if $H_a>H_{c1}$, then the Gibbs energy is lowered by having vortices.}  
One of the differences between the regimes \eqref{regime1} and \eqref{regime2} is that $|\Psi|^2$ is small in the first regime (the bifurcation problem) and large in the second one. 
If a superconductor fills in the entire $\R^2$, then in the second regime, the average magnetic field per unit area, $b\ra 0$, as $h \to h_{c1}$. 
 \bigskip


\subsection{Time-dependent equations}

A number of dynamical versions of the Ginzburg-Landau equations
appear in the literature. Here we list the most commonly
studied and physically relevant.

\medskip

\noindent
{\it Superconductivity}.
In the leading approximation, the evolution of a superconductor
is described by the gradient-flow-type equations for the
Ginzburg-Landau energy
\begin{equation}\label{GES}
\begin{cases}
    \gamma \partial_{t\Phi} \Psi = \COVLAP{A}\Psi + \kappa^2(1 - |\Psi|^2)\Psi, \\
 \sigma \partial_{t,\Phi} A   = -\CURL^*\CURL A  + \Im(\bar{\Psi}\COVGRAD{A}\Psi).
\end{cases}
\end{equation}
Here $\Phi$ is the scalar (electric) potential, $\gamma$ a complex number, and $\sigma$ a two-tensor, and $\partial_{t\Phi}$ is the covariant time derivative $\partial_{t,\Phi}(\Psi, A)  = ((\partial_t + i\Phi)\Psi, \partial_t A + \nabla\Phi)$.
The second equation is Amp\`ere's law,  $\CURL  B=J$, with $J+J_N +J_S,$ where $  J_N= -\sigma (\partial_t A + \nabla\Phi)$ (using Ohm's law) is the normal current associated to the electrons not having formed Cooper pairs, and $J_S = \Im(\bar{\Psi}\COVGRAD{A}\Psi)$, the supercurrent. 

 These equations are called the {\it time-dependent Ginzburg-Landau equations}
or the {\it Gorkov-Eliashberg-Schmidt equations} proposed by Schmid (\cite{Schmidt}) and Gorkov and Eliashberg (\cite{GE})  (earlier versions are due to Bardeen and Stephen and Anderson, Luttinger and Werthamer).

\medskip

\noindent
{\it Particle physics}.
The time-dependent  $U(1)$ Higgs model is described by
\begin{equation}
\la{MHeq}
\bs
   \partial_{t\Phi}^2 \Psi &= \Delta_A \Psi  + \kappa^2 (1-|\Psi|^2) \Psi \\
   \partial_{t\Phi}^2 A &= -\CURL^*\CURL  A +  \Im(\bar{\Psi}\COVGRAD{A}\Psi),
\end{split}
\end{equation}
coupled (covariant) wave equations
describing the $U(1)$-gauge Higgs model of
elementary particle physics
(written here in the {\it temporal gauge}).
Equations~(\ref{MHeq}) are sometimes also
called the {\it Maxwell-Higgs equations}.

\medskip

In what follows, we concentrate on  the Gorkov-Eliashberg-Schmidt equations, \eqref{GES} and, for simplicity of notation,  we use the gauge, in which the scalar potential, $\Phi$, vanishes, $\Phi=0$.

\DETAILS{\medskip
 Multiplying the second equation in \eqref{GES} by $\sigma^{-1}$
and taking $\div$ of the result, we obtain
\begin{equation}\label{Phi-eq}
\Delta \Phi=-\partial_t \div A+\div\sigma^{-1}
[\im(\overline{\Psi}\nabla_A \Psi)- \CURL^2 A].
\end{equation}
This gives an equation for
$\Phi$, which can be easily solved. The solution is determined up to a \textit{harmonic function} on $\R^2$. We \textit{fix the solution} (up to a constant) by demanding that $\Phi$ is bounded.
In what follows we always
assume that $\Phi$ is a bounded solution of the equation \eqref{Phi-eq} and, in
particular, is a function of $\Psi$ and $A$, and we do not list it
among unknowns and use the notation $u=(\Psi, A)$.} 


\section{Vortices}


\subsection{$n$-vortex solutions}

A model for a vortex is given, for each degree $n \in \Z$,
by a ``radially symmetric''
(more precisely {\it equivariant}) solution
of the Ginzburg-Landau equations~\eqref{GLE} of the form
\begin{equation}
\label{eq:vort}
   \Psi^{(n)} (x) = f_n (r) e^{in\theta} {\hbox{\quad and \quad}}
   A^{(n)}(x) = a_n (r) \nabla (n\theta) \ ,
\end{equation}
where $(r,\theta)$ are the polar coordinates of $x \in \R^2$.
Note that $\deg(\Psi^{(n)}) = n$.
The pair $(\Psi^{(n)}, A^{(n)})$ is called the $n$-{\it vortex}
({\it magnetic} or {\it Abrikosov} in the case of
superconductors, and {\it Nielsen-Olesen} or {\it Nambu string} in
the particle physics case). For superconductors, this is a mixed
state with the normal phase residing at the point where the vortex
vanishes. The existence of such solutions of the Ginzburg-Landau
equations was already noticed by Abrikosov~\cite{Abr} and proven in \cite{BC}.  

Using  self-duality, and consequent reduction to a first-order equations, 
Taubes~\cite{Taub1, Taub2} has showed that for a given degree $n$, the family of solutions
modulo gauge transformations ({\it moduli space}) is
$2|n|$-dimensional, and the $2|n|$ parameters
describe the locations of the zeros of the scalar field
-- that is, the vortex centers. A review of this theory
can be found in the book of Jaffe-Taubes~\cite{JT}.

The $n$-vortex solution exhibits the length scales discussed
above. Indeed, the following asymptotics for the
field components of the $n$-vortex~\eqref{eq:vort}
were established in~\cite{p} (see also~\cite{JT}):
\be
\label{eq:decay}
  \ba{c}
  J^{(n)}(x) = n \beta_{n} K_1(r)[1 + o(e^{-m_{\kappa}r})] J\hat{x} \\
  B^{(n)}(r) = n \beta_{n} K_1(r)[1 - \frac{1}{2r} + O(1/r^2)] \\
  |1-f_n(r)| \leq c e^{-m_{\kappa} r},\
  |f_n'(r)| \leq c e^{-m_{\kappa} r},
  \ea
\end{equation}
as $r := |x| \ra \infty$, where
$J^{(n)} := \Im(\overline{\Psi^{(n)}} \nabla_{A^{(n)}} \Psi^{(n)})$
is the $n$-vortex supercurrent,
$B^{(n)} := \curl A^{(n)}$ is the $n$-vortex magnetic field,
$\beta_n > 0$ is a constant, and
$K_1$ is the modified Bessel function of order $1$
of the second kind. The length scale of $\Psi^{(n)}$ is $1/m_\kappa$. Since $K_1(r)$ behaves like
$c e^{-r} / \sqrt{r}$ for large $r$, we see
that the length scale for $J^{(n)}$ and $B^{(n)}$ is $1$. (In fact, for $x\ne 0$, $\Psi^{(n)}$ vanishes as $\kappa\ra \infty$.)




\subsection{Stability} 
\label{sec:stab}



\DETAILS{The linearized stability/instability result of
Theorem~\ref{thm:stab} for the exact $n$-vortex solution
$(\psi^{(n)}, A^{(n)})$ has dynamical counterparts,
proved in~\cite{G}, showing that the $n$-vortex is stable
as a solution of the gradient flow equations~\eqref{eq:ge}
or the Maxwell-Higgs equations~\eqref{eq:mh} for all $n$ in the
Type I case ($\kappa < 1/\sqrt{2}$), and for $n = \pm 1$ in the
Type II case ($\kappa > 1/\sqrt{2}$). Otherwise, the $n$-vortex is\
unstable.

To be more precise, we}
We say the $n$-vortex is {\it (orbitally) stable}, 
if for any initial data sufficiently close to the $n$-vortex (
which includes initial momentum field in the~\eqref{MHeq} case), 
the solution remains, for all time, close to {\it an element of the orbit of
the $n$-vortex under the symmetry group}. Here ``close'' can
be taken to mean close in the ``energy space'' Sobolev norm
$H^1$. 

Similarly, for {\it asymptotic stability} : the solution converges, as
$t \to \infty$, to an element of the symmetry orbit
(that is, to a spatially-translated, gauge-transformed
$n$-vortex).


The basic result on vortex stability is the following:

\begin{theorem}[\cite{GS, G}]
\label{thm:stab}

\begin{enumerate}
\item
For Type I superconductors, all $n$-vortices are asymptotically stable.
\item
For Type II superconductors,  the $\pm1$-vortices are stable, while the $n$-vortices with $|n| \geq 2$, are unstable.
\end{enumerate}
\end{theorem}
This stability behaviour was long conjectured (see~\cite{JT}),
based on numerical computations (eg. \cite{jr}) leading to a ``vortex interaction''
picture wherein inter-vortex interactions are always attractive in
the Type-I case, but become repulsive for like-signed vortices in the
Type-II case.

This result agrees with the fact, mentioned above, that the surface tension is positive for $\kappa < 1/\sqrt{2}$ and negative for $\kappa > 1/\sqrt{2}$, so the vortices try to minimize their 'surface' for $\kappa < 1/\sqrt{2}$ and maximize it for $\kappa > 1/\sqrt{2}$.

 Stability for  pinned vortices was proven  in~\cite{GT}.

For the Maxwell-Higgs equations~\eqref{MHeq}, the above result was proven for the orbital stability only (see \cite{G}). The \textit{asymptotic stability} of the $n$-vortex for these equations is not known.

See also \cite{ComSau, Sau} for extensions of these results to domains other than $\R^2$.

To demonstrate the above theorem, we first prove the linearized/energetic stability or instability. To formulate the latter, we observe that the $n$-vortex is a critical point of the Ginzburg-Landau
energy~\eqref{GLEn}, and the second variation of the
energy
\[
  L^{(n)} := \mbox{ Hess } \E (\Psi^{(n)}, A^{(n)})
\]
is the linearized operator for the Ginzburg-Landau
equations~\eqref{GLE} around the $n$-vortex, acting on
the space $X = L^2(\R^2,\C) \oplus L^2(\R^2,\R^2)$.
 ($\mbox{Hess } \E (u)$ is the G\^{a}teaux derivative of the $L^2-$gradient of the l.h.s. of \eqref{GLE} and $u=(\Psi, A)$. Since $\mbox{Hess } \E (u)$  is only real-linear, to apply the spectral theory,  it is convenient to extend it to a complex-linear operator. However, in order not to introduce extra notation, we ignore this point here and {\it deal with $L^{(n)}$ as if it were a complex-linear operator}.) 

The symmetry group of $\E(\Psi,A)$, which is infinite-dimensional
due to gauge transformations, gives rise to an infinite-dimensional
subspace of $\Null(L^{(n)}) \subset X$, which we denote here by
$Z_{sym}$.
We say the $n$-vortex is {\em (linearly) stable} if for some $c > 0$,
\begin{equation}\label{coercL}
  L^{(n)}|_{Z_{sym}^{\perp}} \geq c>0,
\end{equation} 
and {\em unstable} if $L^{(n)}$ has a negative eigenvalue.
By this definition, a stable state is a local energy minimizer which
is a {\it strict} minimizer in directions orthogonal to the
infinitesimal symmetry transformations. An unstable state
is an energy saddle point.

Once  the linearized (spectral)  stability is proven,  the main task in proving the
orbital stability is the construction of a path in the
(infinite dimensional, due to gauge symmetry) symmetry group
orbit of the $n$-vortex, to which the solution remains close.
For the gradient-flow equations~\eqref{GES} the
orbital stability can be easily strengthened to and with little more work, the asymptotic stability can be accomplished.

A few brief remarks on the proof of the key step \eqref{coercL}:

\paragraph{$\bullet$} 
Since the vortices are gauge equivalent under the action of rotation, i.e.,
$$\Psi(R_\alpha x) = e^{in\alpha}\Psi(x),\ R_{-\alpha}A(R_\alpha x) = A(x),$$
where $R_\al$ is counterclockwise rotation in $\R^2$ through the angle $\al$,
the linearized operator $L^{(n)}$ commutes with  
the representation $\rho_n : U(1) \rightarrow Aut([L^2({\R}^2; {\C})]^4)$ of the group $U(1)$,  given by
\[
  \rho_n(e^{i\theta})
  (\xi,  \al)(x) =  (e^{in\theta}\xi, e^{-i\theta}\al) (R_{-\theta}x).
\]
It follows that $L^{(n)}$ leaves invariant the eigenspaces of $d\rho_n(s)$ for any $s \in i{\R} = Lie(U(1))$. (The representation of $U(1)$ on each of these subspaces is multiple to an irreducible one.)
 According to a representation of the
symmetry group, this results in (fiber) block decomposition of $L^{(n)}$, 
\begin{align} \label{X-deco}  X \approx \bigoplus_{m \in \Z} ( L^2_{rad} )^4,
  \quad\quad L^{(n)} \approx \oplus_{m \in \Z} L^{(n)}_m, \end{align} 
  where $L^2_{rad} \equiv L^2({\R}^{+}, rdr)$ and $\approx$ stands for the unitary equivalence,  which is described below.  
One can then study each operator $L^{(n)}_m$, which acts
on (vectors of) radially-symmetric functions.

\paragraph{$\bullet$} 
The gauge adjusted translational zero-modes each lie within a single subspace of the decomposition \ref{X-deco} and correspond, after complexification  and rotation, to the vector
\begin{align} \label{transl-mode}  T = ( \; f_n'(r), \; b^{(n)}(r) f_n(r), \; n a_n'(r)/r,
  \; n a_n'(r)/r \; ), \end{align}   where we defined, for convenience, $b^{(n)}(r) = \frac{n(1-a^{(n)}(r))}{r}$, in the $m = \pm 1$ sectors. The stability proof is built on
Perron-Frobenius-type arguments (``positive zero-mode $\iff$
the bottom of the spectrum''), adapted to the setting of systems.
In particular, $T$ lies in the ``positivity cone'' of
vector-functions with positive components, and we are able to
conclude that
$L^{(n)}_{\pm 1} \geq 0$ with non-degenerate zero-eigenvalue.
\paragraph{$\bullet$} A key component is the exploitation of the special structure,
hinted at by the Bogomolnyi lower bound~\eqref{bog},
of the linearized operator $L^{(n)}$ at the self-dual
value $\kappa = 1/2$ of the Ginzburg-Landau parameter. In fact,
\[
  L_m^{(n)} |_{\kappa = 1/2} = (F^{(n)}_m)^* F^{(n)}_m
\]
for a first-order operator, $F_m$, having $2|n|$ zero-modes
which can be calculated semi-explicitly. These modes
can be thought of as arising from independent relative
motions of vortices, and the fact that they are energy-neutral,
relates to the vanishing of the vortex interaction at $\kappa=1/2$
\cite{Bog,Wein}. Two of the modes arise from translational
symmetry, while careful analysis shows that as
$\kappa$ moves above (resp. below) $1/2$, the
$2|n|-2$ ``extra'' modes become unstable (resp. stable)
directions.


\paragraph{} Technically, it is convenient, on the first step, effectively remove the (infinite-dimensional subspace of) gauge-symmetry
zero-modes, by modifying $L^{(n)}$ to make it coercive in
the gauge directions -- this leaves only the two zero-modes
arising from translational invariance remaining.

\DETAILS{  In what follows we write functions on ${\R}^2$ in polar coordinates, so that
\begin{equation} \label{polar}
  \cH^c(\R^2): = [L^2({\R}^2; {\C})]^4 =  [L^2_{rad} \otimes
  L^2({\bf S}^1; {\C})]^4.
\end{equation}}
 Let $\mathcal{C}$ be the operation of taking the complex conjugate. The results in (fiber) block decomposition of $L^{(n)}$, mentioned above is given in 
 \begin{theorem}[\cite{GS}]\label{thm:decomp}
\begin{enumerate}[(a)]
\item Let $\cH_m := [L^2_{rad} ]^4$ and define $U : 
X \to \cH$, where $ \cH = \bigoplus_{m \in {\bf Z}}\cH_m$, 
so that on smooth compactly supported $v$ it acts by the formula
	\begin{equation*}
		(U v)_m(r) = J_m^{-1} \int_0^{2\pi} \chi_m^{-1}(\theta) \rho_n(e^{i\theta}) v(x) d\theta.
	\end{equation*}
where $\chi_m(\theta)$ are characters of $U(1)$, i.e., all homomorphisms $U(1) \to U(1)$ (explicitly we have $\chi_m(\theta)=e^{im\theta}$) and
\[J_m:\cH_m \ra e^{i(m+n)\theta} L^2_{rad} \oplus   e^{i(m-n)\theta} L^2_{rad} \oplus -i e^{i(m-1)\theta} L^2_{rad} \oplus  i e^{i(m+1)\theta} L^2_{rad} \]
 acting in the obvious way.
Then $U$ extends uniquely to a unitary operator.
	
\item Under $U$ the linearized operator around the vortex, $K_\#^{(n)}$,
decomposes as
\begin{equation}\label{eq:flip}
  U L^{(n)} U^{-1}   = \bigoplus_{m \in {\bf Z}} L_m^{(n)},
\end{equation}
where the operators $ L_m^{(n)}$ act on $\cH_m$ as $J_m^{-1} L^{(n)} J_m$.

\item The operators $K_m^{(n)}$ have the following properties:
\DETAILS{passing to a rotated version, $M_m^{(n)}$,
of the operator $K_m^{(n)}$,
\[
  M_m^{(n)} \equiv \left\{ \begin{array}{cc}
            R K_m^{(n)} R^T  &  m \geq 0  \\
            R' K_m^{(n)} (R')^T  &  m < 0
            \end{array} \right.
\]
where
\[
  R = \frac{1}{\sqrt{2}} \left( \begin{array}{cccc}
      1 & 1 & 0 & 0 \\
      -1 & 1 & 0 & 0 \\
      0 & 0 & 1 & 1 \\
      0 & 0 & 1 & -1  \end{array} \right),    \;\;\;\;\;\;\;\;\;\;
  R' = \frac{1}{\sqrt{2}} \left( \begin{array}{cccc}
       1 & 1 & 0 & 0 \\
       1 & -1 & 0 & 0 \\
       0 & 0 & 1 & 1 \\
       0 & 0 & 1 & -1  \end{array} \right),
\]
we have}
%
\begin{equation}\label{eq:flip}
  K_m^{(n)} = R K_{-m}^{(n)}R^T,\ \mbox{where}\  R= \TwoByTwo{Q}{0}{0}{Q}, Q = \TwoByTwo{0}{\cC}{\cC}{0},
\end{equation}
\begin{equation}\label{eq:contspec}
  \sigma_{ess}(K_m^{(n)}) = [\min(1,\lambda), \infty),
\end{equation}
\begin{equation}\label{eq:mon}
 \mbox{for}\ |n|=1\ \mbox{and}\  m \geq 2,\ \, L_m^{(n)} - L_1^{(n)} \geq 0\  \, \mbox{with no zero-eigenvalue,}
\end{equation}
\begin{equation}\label{eq:mon}
L_0^{(n)} \geq c > 0\ \quad \mbox{for all}\ \kappa,
\end{equation}
\begin{equation}
\label{L1}
L_1^{(\pm1)} \geq 0 \  \mbox{with non-degenerate zero-mode given by \eqref{transl-mode}.}
\end{equation}
\end{enumerate}
\end{theorem}

Since, by ~(\ref{eq:contspec}) and \eqref{L1},
$ L_1^{(\pm1)}|_{T^{\perp}} \geq \tilde{c} > 0$ and,
 by~(\ref{eq:mon}) and \eqref{L1}, $ L_m^{(\pm1)} \geq c' > 0$
for $|m| \geq 2$, this theorem implies \eqref{coercL}. 
\DETAILS{\begin{corollary}[coercivity of $K_\#$] \la{cor:coercK}
 On the subspace of $\cH^c(\R^2)$ orthogonal
to the translational zero-modes, we have for some $c > 0$,
\begin{equation}\label{coercK}
 K_\# \geq c > 0.
\end{equation}
\end{corollary}}


\section{
Vortex Lattices}

In this section we describe briefly recent results on vortex lattice solutions, i.e. solutions, which display vortices arranged along vertices of a lattice in $\R^2$. 
Since their discovery by Abrikosov in 1957, solutions have been studied in numerous experimental and theoretical works (of the more
mathematical studies, we mention the articles of Eilenberger~\cite{Eil} and Lasher~\cite{Lash}).

The rigorous investigation of Abrikosov solutions began in \cite{Odeh},  soon after their discovery. Odeh has given a detailed sketch of the proof the bifurcation of  Abrikosov solutions at the second critical magnetic field. 
Further details were provided by Barany, Golubitsky, and Tursky \cite{BGT}, 
using equivariant bifurcation theory, and by Tak\'{a}\u{c} \cite{Takac}, 
who obtained results on the zeros of the bifurcating solutions. The proof of existence was completed in \cite{TS} and extended further in \cite{TS2} beyond the cases covered in the works above.

Existence of  Abrikosov solutions at low magnetic fields, near the first  critical magnetic field was given in \cite{ST1}.

Moreover, Odeh has also given a detailed sketch of the  proof, with details filled in in  \cite{Dutour2}, of the existence of  Abrikosov solutions using the variational  minimization of the Ginzburg-Landau energy functional reduced to a fundamental cell of the underlying lattice. However, this proof provides only very limited information about the solutions. 

Moreover,  important and fairly detailed results on asymptotic behaviour of solutions, for $\kappa \ra\infty$ and the applied magnetic fields, $h$, satisfying $h\le \frac{1}{2}\log \kappa+$const (the London limit), were obtained in \cite{AS} (see this paper and the book \cite{SS} for references to earlier works). Further extensions to the Ginzburg-Landau equations for anisotropic and high temperature superconductors can be found in \cite{ABS1, ABS2}.

Among related results, a relation of the Ginzburg-Landau minimization problem, for a fixed, finite domain and in the regime of the Ginzburg-Landau parameter $\kappa\ra\infty$ and external magnetic field, to the Abrikosov lattice variational problem was obtained in \cite{AS, Al2}.   \cite{Dutour}  (see also \cite{Dutour2}) have found boundaries between superconducting, normal and mixed phases.

The proof that the triangular lattices minimize  the Ginzburg-Landau energy functional per the fundamental cell was completed in \cite{TS} using the original   Abrikosov ideas and results on the   Abrikosov 'constant' due to  \cite{ABN, NV}.

The stability of  Abrikosov lattices was shown in  \cite{ST2} for gauge periodic perturbations, i.e. perturbations having the same translational lattice symmetry as the solutions themselves, and in   \cite{ST3} for local, more precisely, $H^1$, perturbations.

Here we describe briefly the existence 
and stability results and the main ingredient entering into their proofs. 



\subsection{Abrikosov lattices}\label{sec:abr-lat}

In 1957,  A. Abrikosov (\cite{Abr}) discovered a class of solutions, $(\Psi, A)$, to \eqref{GLE}, presently known as Abrikosov lattice vortex states (or just Abrikosov lattices), whose physical characteristics, density of  Cooper pairs,  $|\Psi|^2$, the magnetic field,  $\CURL A$,  and the supercurrent,  $J_S = \Im(\bar{\Psi}\COVGRAD{A}\Psi)$, are double-periodic w.r. to a lattice $\LAT$.
(This set of states is invariant under  the symmetries of the previous subsection.) 




For Abrikosov states, for $(\Psi, A)$,  the magnetic flux, $
\int_\Omega \Curl A$, through a lattice cell, $\Om$, is quantized,
\begin{equation}\label{flux-quant}
   \frac{1}{2\pi} \int_\Omega \Curl A =\deg \Psi =  n,
\end{equation}
for some integer $n$.  Indeed, the periodicity of $n_s=|\Psi| ^2$ and $J=\im(\bar \Psi \n_A \Psi)$ imply that $\nabla\varphi - A$, where $\Psi = |\Psi|e^{i\varphi}$, is periodic, provided $\Psi \neq 0$ on $\partial\Omega$.  This, together with Stokes's theorem, $\int_\Omega \Curl A = \oint_{\partial\Omega} A = \oint_{\partial\Omega} \nabla\varphi$ and the single-valuedness of $\Psi$,   imply that $\int_\Omega \Curl A =2\pi n$ for some integer $n$.  
 Using the reflection symmetry of the problem, one can easily check that we can always assume $n \geq 0$.

\DETAILS{Note that if $|\psi| > 0$ on $\partial \Omega$, we may write
$\psi = |\psi|e^{i \phi}$ there, and then the fact that
$J = Im(\bar{\psi}\nabla_A \psi) = |\psi|^2(\nabla \phi - A)$
satisfies periodic boundary conditions on $\partial \Omega$ implies
that $\nabla \phi - A$ does as well, and we are led to the \
flux quantization relation
\[
  \int_{\Omega} curl A
  = \int_{\p \Omega} A \cdot \tau
  = \int_{\p \Omega} \nabla \phi \cdot \tau
  := 2 \pi n \in 2 \pi \Z.
\]}

 Equation \eqref{flux-quant} implies the relation between the average magnetic flux, $b$, per lattice cell, $b = \frac{1}{|\Omega|} \int_\Omega  \CURL A $,  and the area, $|\Om|$, of a fundamental cell
\begin{equation}\label{quant-cond}
     b = \frac{2\pi n}{|\Omega|}.
\end{equation}

Finally, it is clear that the gauge, translation, and rotation symmetries of
the Ginzburg-Landau equations map lattice states to
lattice states. In the case of the gauge and translation symmetries, the lattice with respect to which the solution is
gauge-periodic does not change, whereas with the rotation symmetry, the lattice is rotated as well. The magnetic flux per cell of solutions 
is also preserved under the action of these symmetries.


\DETAILS{\subsection{Energy of lattice states}
Lattice states have infinite total energy, so we consider
the average energy per cell,
\begin{equation*}
    \mathcal{F}(\psi, A) := \frac{1}{|\Omega|} \mathcal{E}_\Omega(\psi, A).
\end{equation*}
Here, $\Omega$ is a primitive cell of the lattice $\L$,
and $|\Omega|$ is its Lebesgue measure.}


\subsection{Existence of Abrikosov lattices}

\DETAILS{We identify $\R^2$ with $\C$, via the map $(x_1, x_2)\ra x_1+i x_2$ and  use translations and  rotations, if necessary, to bringing a lattice $\LAT$, satisfying the quantization condition \eqref{quant-cond},  into the form $$\LAT=\sqrt{\frac{2\pi n}{\im\tau b} }  (\Z+\tau\Z),$$ where $\tau\in \C$  is a complex number with $\Im \tau>0$, called the lattice shape parameter. 

Thus,}
We assume always that  the co-ordinate origin is placed at one of the vertices of the lattice $\LAT$. By the shape of  lattice $\LAT$ we understand the equivalence class $[\LAT]$  of lattices, with equivalence relations given by   rotations and dilatations. 
We will show in Appendix \ref{sec:param} that  lattice shapes can be parametrized by  points $\tau$ in  the fundamental domain,  $\Pi^+/SL(2, \Z)$, of  the modular group  $SL(2, \Z)$ acting on the Poicar\'e half-plane $\Pi^+:=\{\tau \in \C: \Im \tau>0\}$ (see Fig. \ref{fig:PoincareStrip}).
(We denote the corresponding equivalence class by $[\tau]$.)

\begin{figure}[h!]
	\centering 

  \includegraphics[width=2.5in]{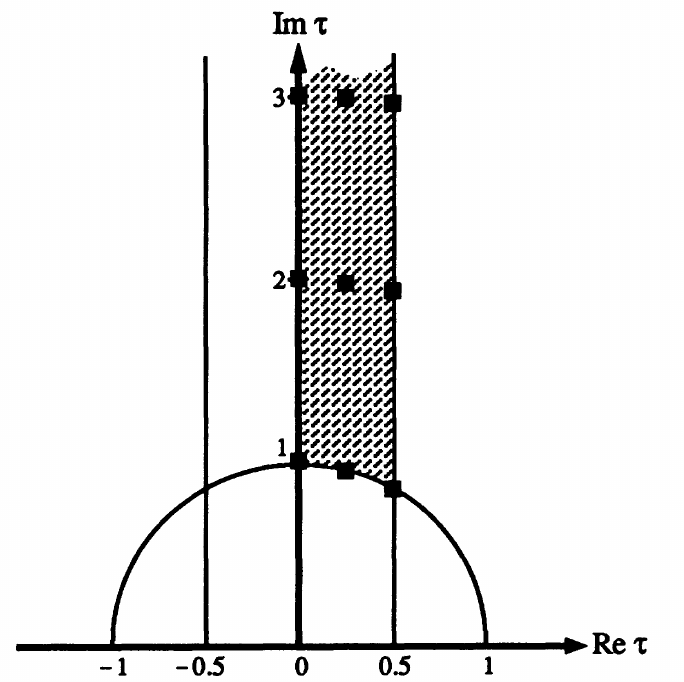} 
  \caption{Fundamental domain  of $\gamma(\tau)$. }\label{fig:PoincareStrip}
\end{figure}  

Due to the quantization relation \eqref{quant-cond}, the parameters $\tau$, $b$, and $n$ determine the lattice $\LAT$ up  to a rotation and a translation. 
As the equations \eqref{GLE} are invariant under rotations and translations, solutions corresponding to translated and rotated lattices are related by 
 symmetry transformations and therefore can be considered equivalent, with equivalence classes 
 determined by triples $\om=(\tau, b, n)$, specifying the underlying lattice has shape $\tau$, the average magnetic flux per lattice cell  $b$, and the number $n$ of quanta of magnetic flux per lattice cell.  With this in mind,    we will say that an 
Abrikosov lattice state $(\Psi, A)$ is of type $\om=(\tau, b, n)$, if it belong to the equivalence class determined by $\om=(\tau, b, n)$. 


Let $\beta(\tau)$  be  the Abrikosov 'constant', defined in \eqref{beta} below. The following critical value of the Ginzburg-Landau parameter $\kappa$ plays an important role in what follows 
\begin{equation}\label{kappac}
     \kappa_c(\tau) := \sqrt{\frac{1}{2}\left(1-\frac{1}{\beta (\tau)}\right)} .
     \end{equation}
Recall that the value of the second critical magnetic field at which the normal material undergoes the transition to the
superconducting state is  that $h_{c2}=\kappa^2$.

For the case $n=1$ of one quantum of flux per unit cell,
the following result establishes the existence of non-trivial
lattice solutions near the normal metal solution:

\begin{theorem}[\cite{Odeh, BGT, Dutour2, TS2}] \label{thm:ALexist1}
Fix a lattice shape $\taul$ and let $b$ satisfy 
\begin{equation}\label{b-cond} |\kappa^2 - b|\ll \kappa^2[(2\kappa^2 - 1)\beta(\tau) + 1] \end{equation}  and   
\begin{equation}\label{b-cond2} \text{either}\ \kappa> \kappa_c(\tau),\ \kappa^2 > b\ \text{  or  }\  \kappa < \kappa_c(\tau),\ \kappa^2 < b.\end{equation}   Then for  $\om =(\tau, b, 1)$  
\begin{itemize} 
\item      there exists a smooth Abrikosov lattice solution $u_\om = (\Psi_\om, A_\om)$ of type $\om$. 
 \end{itemize}     \end{theorem}
\medskip

\noindent {\bf Remark.}
  For  $\kappa> \frac{1}{\sqrt 2}$ and the triangular and square lattices the theorem was proven in \cite{Odeh, BGT, Dutour2, TS} and in the case stated, in  \cite{TS, TS2}.

\begin{theorem}[\cite{TS2}] \label{thm:ALener}
Let  $\kappa> \frac{1}{\sqrt 2}$. Fix a lattice shape $\tau$ and let  
$b$ satisfy   $b< \kappa^2$  and \eqref{b-cond}. Then 
%
 \begin{itemize} 
\item   the global minimizer of the average energy per cell is the solution
    corresponding to the equilateral triangular lattice.
\end{itemize}    \end{theorem}
(Due to a calculation error, Abrikosov concluded that the
lattice which gives the minimum energy is the square lattice.
The error was corrected by Kleiner, Roth, and Autler~\cite{KRA},
who showed that it is in fact the triangular lattice which
minimizes the energy.)

Now, we formulate the existence result for low magnetic fields, those near the first critical magnetic field $h_{c1}$: Let $\cL_{\om}$ be a lattice specified by a triple $\om=(\tau, b, n)$ and let  $\Om_\om$ denote its elementary cell. We have 
\DETAILS{Let  $\cL\equiv\cL_{\taul \rho}$ be a family of lattices of a fixed shape, $\taul$, and with $\rho$ being the length of the shortest basis vector for $\cL_{\tau \rho}$. 
We assume  $\rho\gg 1$.  Then, by  the flux quantization,	$b = \frac{2\pi n}{|\Omega|},$ where $b$ is the average magnetic flux per the fundamental cell, the area of the fundamental cell, $\Om_\om$, of $\cL$ is $\ge \rho^2$ and the average magnetic field $b=O(\rho^{-2})$.} 
\begin{theorem}[\cite{ST1}] \label{thm:ALexist2}
Let $\kappa\ne \frac{1}{\sqrt{2}}$ and fix a lattice shape $\lam$ 
 and $n \ne 0$. 
Then there is  $b_0=b_0(\kappa )\ (\sim (\kappa- 1/\sqrt{2})^{2}) > 0$ such that for $b  \le b_0$, there exists an odd solution Abrikosov lattice solution $u_{\om} \equiv (\Psi_{\om}, A_{\om})$ of~\eqref{GLE}, 
s.t. 
\be  \la{eq:close}
  u_{\om} (x)= u^{(n)}(x-\alpha)+O(e^{-c \rho})\ \text{  on  }\ \Omega_\om+\alpha,\ \forall \alpha \in \cL_{\om}, \end{equation}
where $u^{(n)}:=(\Psi^{(n)}, A^{(n)})$ is the $n-$vortex, $\rho=b^{-1/2}$ and $c>0$, in the sense of the local Sobolev norm of any index.
\end{theorem}
In the next two subsections we present a discussion of some key general notions. After this, we outline the proofs of the results above.

\subsection{Abrikosov lattices as gauge-equivariant states}\label{sec:equiv}
A key point in proving both theorems is to realize that a state $(\Psi, A)$ is an Abrikosov lattice if and only if  $(\Psi, A)$ is gauge-periodic or  gauge-equivariant (with respect to a lattice $\LAT$) in the sense that
there exist (possibly multivalued) functions $g_s : \R^2 \to \R$, $s \in \LAT$, such that
\begin{equation}\label{gauge-per'}
	T^{\rm trans}_s (\Psi, A) =T^{\rm gauge}_{g_s} (\Psi, A).
\end{equation}
Indeed, if state  $(\Psi, A)$ satisfies \eqref{gauge-per'}, then all associated physical quantities are $\cL-$periodic, i.e. $(\Psi, A)$ is an Abrikosov lattice. 
In the opposite direction, if $(\Psi, A)$ is an Abrikosov lattice, then $\curl A(x)$ is periodic w.r.to $\cL$,  and therefore   $A(x + s) = A(x) +\n g_s(x)$,  for some functions $g_s(x)$. 
Next, we write $\Psi(x)=|\Psi(x)|e^{i \phi(x)}$. Since $|\Psi(x)|$ and $J(x)= |\Psi(x)|^2 (\n \phi(x)-  A(x))$ are  periodic w.r.to $\cL$,  we have that   $\n \phi(x + s) = \n \phi(x) +\n \tilde g_s(x)$, which implies that   $\phi(x + s) =\phi(x) + g_s(x)$, where $g_s(x)=\tilde g_s(x)+ c_s$, for some constants $c_s$. 

Since $T^{\rm trans}_s$ is a commutative group, we see that  the family of functions $g_s$ has the important cocycle property
\begin{equation}\label{cocycle-cond}
    g_{s+t}(x) - g_s(x+t) - g_t(x) \in 2\pi\Z.
\end{equation}
This can be seen by evaluating the effect of translation by $s+t$ in two different ways. We call $g_s(x)$ the {\it gauge exponent}. 
It can be shown (see Appendix \ref{sec:gs}) that by a gauge transformation, we can pass from a  exponential  $g_s$ satisfying the cocycle condition \eqref{cocycle-cond}   is equivalent to  the exponent 
\begin{equation}\label{gs-spec'} \frac{b}{2} s \wedge x + c_s,\ \text{  with  $c_s$   satisfying }\ c_{s+t} - c_s - c_t - \frac{1}{2} b s \wedge t \in 2\pi\Z,\end{equation} for 
and  $b$ 
  satisfying  $b |\Omega| \in 2\pi\Z$. 
 For more discussion of $g_s$, see Appendix \ref{sec:gs}. 
 

\medskip
\noindent {\bf Remark.} 
 Relation \eqref{cocycle-cond} for Abrikosov lattices was isolated in \cite{ST2}, where it played an important role. This condition is well known in algebraic geometry and number theory (see e.g. \cite{Gun2}).  However, there the associated vector potential (connection on the corresponding principal bundle) $A$ is not considered there.



\subsection{Abrikosov function} \label{sec:ABR}

Let the number of the magnetic flux quant per the lattice cell be $n=1$. Let $\lan f \ran_\Om$  denote the average, $\lan f \ran_\Om = \frac{1}{|\Omega|} \int_{\Omega} f ,$ of a function $f$ over  $\Omega\subset \R^2$. For a lattice $\LAT$, considered as a  group of lattice translations, $\hat\LAT$ denotes the dual group, i.e. the group of characters, $\chi : \LAT \to U(1)$. Furthermore,  let  $\cT:= \R^2/\LAT$ and for $\LAT\in [\tau]$, we introduce the normalized lattice  $\LAT_\tau =\sqrt\frac{2\pi}{|\cT|}\LAT$. 
 The key role in understanding the energetics of is played by the Abrikosov function, 
\begin{equation} \label{beta}
    \beta(\tau) := \frac{ \lan  |\phi|^4\ran_{\Omega_\tau}}{  \lan  |\phi|^2 \ran_{\Omega_\tau}^2 },
\end{equation}
where   $\Om_\tau$ is a fundamental cell of the lattice  $\LAT_\tau$, 
 and  $\phi$ is 
 the solution to the problem
\begin{align}\label{phi-eqs-resc} (-\COVLAP{A^n}-n)\phi =0, \ \quad \phi (x+s) =e^{i\frac n2  s \wedge x} 
\phi (x),\ \quad \forall s\in \LAT_\tau,	\end{align}
where $A^n(x):= -\frac n2  Jx$ and $J:= \left( \begin{array}{cc} 0 & 1 \\ -1 & 0 \end{array} \right)$, for $n=1$. 
 We will show below (Proposition \ref{prop:null-set}) that, for $n=1$, the problem \eqref{phi-eqs-resc} has a unique solution and therefore 
 $\beta$ is well-defined. It is not hard to see that $\beta$ depends only on the equivalence class of $\LAT$.  

 Definition \eqref{beta} - \eqref{phi-eqs-resc} implies that $\beta( \tau)$ is  symmetric w.r.to the imaginary axis, $\beta (- \bar\tau)= \beta( \tau)$.  Hence it suffices to consider $ \beta( \tau)$ on the $\Re\tau \geq 0$ half of  the fundamental domain,  $\Pi^+/SL(2, \Z)$ (the heavily shaded area on  Fig. \ref{fig:PoincareStrip}).

Moreover,  we can consider \eqref{beta} - \eqref{phi-eqs-resc} on the entire  Poicar\'e half-plane $\Pi^+:=\{\tau \in \C: \Im \tau>0\}$, which allows us to define  $\beta(\tau)$ as a modular  function on $\Pi^+$, 
 \begin{itemize} \item  the function $\beta(\tau)$,  defined on $\Pi^+$,  is invariant under the action of 
 $SL(2, \Z)$.  \end{itemize} 
\DETAILS{\begin{proof} Clearly,  the function  $\beta(\tau)$ is independent of  should not depend on a choice of a basis in $\LAT_\tau$. 
Now, the parameter $\tau$ in the formulae above determines the basis   $\sqrt{\frac{2\pi}{\im\tau} }(1 ,\  \tau)$  of the lattice $\LAT_\tau$ and any two bases in $\LAT_\tau$  are related by a modular map  $\sqrt{\frac{2\pi}{\im\tau} } (\nu_1,\ \nu_2)\ra  \sqrt{\frac{2\pi}{\im\tau} }(\nu_1',\ \nu_2')=\sqrt{\frac{2\pi}{\im\tau} }(\al\nu_1+\beta,\ \g\nu_2 +\del)$, where  $\al, \beta, \g, \del\in \Z$,  and 
$\al \del-\beta \g=1$ (with the matrix  $\left( \begin{array}{cc} \al & \beta \\ \g & \del \end{array} \right)$, an element  of the modular group  $SL(2, \Z)$). In particular, replacing the basis  $\sqrt{\frac{2\pi}{\im\tau} }(1, \tau)$  by the basis $\sqrt{\frac{2\pi}{\im\tau} }( \g \tau+\del,\  \al \tau+\beta )$, results in the shape parameter $\tau$ being mapped as $ \tau\ra g\tau:=\frac{\al\tau+\beta}{\g\tau+\del}$.
  Hence the statement follows.\end{proof}} 

This implies that it suffices to consider $\beta(\tau)$ on  the fundamental domain \eqref{fund-domSL2Z}.

\medskip

\noindent
{\bf Remarks}.
\DETAILS{1) {\bf For $n>1$ one can define the Abrikosov function by the same formula \eqref{beta} {\bf where} $\phi$ should be thought of as the leading approximation to the Abrikosov solution $\Psi_\om$. We conjecture that}
\begin{equation}
    \beta(\tau) = \inf_{\phi} \frac{ \lan  |\phi|^4\ran_{\Omega_\tau}  }{  \lan  |\phi|^2 \ran_{\Omega_\tau}^2 },
\end{equation}
the $\inf$ taken over a suitable class of lattice states.}
\DETAILS{\item All the above results dealt with Abrikosov lattices with one quantum of magnetic flux per lattice cell. Existence of Abrikosov lattices with more than one quantum of magnetic flux per lattice cell
is shown in \cite{OST}. Earlier partial results in this direction
were obtained in \cite{Abr, Ch, Al}.}
1) The term   Abrikosov constant  comes from the physics literature, where one often considers only equilateral triangular or square lattices.

2)  The way we defined the Abrikosov constant $\beta(\tau)$, it  is manifestly independent of $b$. Our definition differs from the standard one by rescaling: the standard definition uses the function   $\phi_b(x) = \phi(\sqrt b x)$, instead of  $\phi(x) $.

\subsection{Comments on the proof of 
Theorem \ref{thm:ALexist1}} \label{sec:approach-exist1}
As was mentioned in Subsection \ref{sec:equiv}, we look for solutions of \eqref{GES}, satisfying  condition \eqref{gauge-per'}, or explicitly as, for  $s \in \LAT$, 
  \begin{equation}\label{gauge-per}
\begin{cases}	\Psi(x+s) = e^{i g_s(x)}\Psi(x),\\
A(x+s)=A(x) + \n g_s(x),
\end{cases}\end{equation}
where $g_s$ satisfies  \eqref{cocycle-cond}. By \eqref{gs-spec'}, 
it can be taken to be 
\begin{equation}\label{gs-fixed}		g_s(x)  = \frac{b}{2} s \wedge x + c_s,
	\end{equation}	
	where $b$ is the average magnetic flux, $b = \frac{1}{|\Omega|} \int_\Omega  \CURL A $ (satisfying \eqref{quant-cond} so that $b s \wedge t \in 2\pi\Z$),  and the $c_s$ satisfy
\begin{equation}\label{cs-condition}	c_{s+t} - c_s - c_t - \frac{1}{2} b s \wedge t \in 2\pi\Z. \end{equation}

\paragraph{The linearized problem.} We expect that as the average flux $b$ decreases below $h_{c2} = \kappa^2$,
a vortex lattice solution 
emerges from the normal material solution $(\Psi_n , A_n)$, where $\Psi_n = 0$ and  $A_n$ is a magnetic potential, with the constant magnetic field $b$. 
Note that $(\Psi_n , A_n)=(0 , A^b)$ satisfies \eqref{gauge-per}, if we take the gauge  $A^b=  - \frac b2  J x$. 
 Linearizing \eqref{GLE} at  $(0 , A^b)$,  
leads to the  linearized problem
\begin{align}\label{phi-eqs} (-\COVLAP{A^b}-\kappa^2)\phi =0, 
	\end{align}
 with 
$\phi (x)$ satisfying 
\begin{align}\label{phi-per} \phi (x+s) =e^{
i\frac b2 s\cdot J x} \phi (x),\ \quad \forall s\in \LAT.
	\end{align}
(The second equation in \eqref{GLE} leads to  $\curl a=0$ which gives, modulo gauge transformation, $a=0$.)
We show that this problem has $n$ linearly independent solutions, provided	 $b|\Om|=2\pi n$ and  $b=\kappa^2=h_{c2}$.

Denote by $L^b$ the operator $-\Delta_{A^b}$, defined on the lattice cell $\Omega$ with the lattice
boundary conditions in~\eqref{phi-per}, is self-adjoint, has a purely discrete
spectrum, and evidently satisfies $L^b \ge 0$. We have
\begin{proposition} \label{prop:null-set}
The operator $L^b$ is self-adjoint, with the purely discrete spectrum given by
 the spectrum explicitly as 
	\begin{equation}\label{specLb}
	\sigma(L^b) = \{\, (2k + 1)b : k = 0, 1, 2, \ldots \,\},
	\end{equation}
	and each eigenvalue is of the same multiplicity. 
	
	If $b|\Om|=2\pi n$, then this multiplicity is $n$ and, in particular, we have $$\dim_\C \Null (L^b - b) = n.$$
\end{proposition}
\begin{proof} The self-adjointness is standard. Spectral information about $L^b$ can be obtained by
introducing the harmonic oscillator annihilation and creation  operators, $\alpha $ and $\alpha^*$, with 
    \begin{equation}
        \alpha  := (\nabla_{A^b})_1 + i(\nabla_{A^b})_2 =\partial_{x_1} + i\partial_{x_2} + \frac{1}{2} b x_1 + \frac{1}{2} i b x_2.
    \end{equation}
    One can verify that these operators satisfy the following relations:
    \begin{enumerate}
    \item $[\alpha, \alpha^*] = 2\Curl A^b =2b$;
    \item $-\Delta_{A^b} - b = \alpha^*\alpha$.
    \end{enumerate}
    As for the harmonic oscillator (see for example \cite{GS1}), this gives the spectrum explicitly, \eqref{specLb}. This proves the first part of the theorem. 
    


For the second part, a simple calculation gives the following operator equation
\begin{equation*}
  e^{\frac{b}{2}(i x_1 x_2-x_2^2)}\alpha e^{-\frac{b}{2}(i x_1 x_2-x_2^2)} = \partial_{x_1} + i\partial_{x_2}.
\end{equation*}
This immediately proves that $\psi \in \Null \alpha$ if and only if $\xi (x)= e^{\frac{b}{2}(i x_1 x_2-x_2^2)}\psi (x)$ satisfies $\partial_{x_1}\xi + i\partial_{x_2}\xi = 0$.

We identify $\R^2$ with $\C$, via the map $(x_1, x_2)\ra x_1+i x_2$. 
We can choose a basis in $\LAT$  so that $\LAT=r  (\Z+\tau\Z)$, where $\tau\in \C$, $\Im \tau>0$, and $r>0$. By   the quantization condition \eqref{quant-cond}, $r:=\sqrt{\frac{2\pi n}{\im\tau b} }$. Define $z=\frac{1}{r} (x_1+i x_2)$ and 
\begin{align}\label{theta-phi} 
 \theta(z)= e^{\frac{b}{2}(i x_1 x_2-x_2^2)}\phi (x).
\end{align} 
By the above,  the function $\theta$ is entire and, due to the periodicity  conditions on $\phi$, satisfies
\begin{subequations}
            \begin{equation*}
                \theta(z + 1) = \theta(z),
            \end{equation*}
            \begin{equation*}
 \theta(z + \tau) =  e^{ -2inz } e^{ -in\tau z } \theta(z).
            \end{equation*}
\end{subequations}
Hence $\theta$ is the theta function and has the 
absolutely convergent Fourier expansion
        \begin{align}\label{theta-series}            \theta(z) = \sum_{k=-\infty}^{\infty} c_k e^{2kiz}.
        \end{align}
with the coefficients  satisfying 
      $      c_{k + n} = e^{in\pi\tau} e^{2ki\pi\tau} c_k,$ 
which means such functions are determined by
$c_0,\ldots,c_{n-1}$ and therefore form an $n$-dimensional
vector space. This proves Proposition~\ref{prop:null-set}.
\end{proof} 
This 
also gives the form of the leading approximation \eqref{theta-phi} -  \eqref{theta-series} to the true solution.

\paragraph{The nonlinear problem.}
Now let $n=1$.
Once the linearized map is well understood, it is possible
to construct solutions, $u_\om,\ \om=(\tau, b, 1)$, of the Ginzburg-Landau equations
for a given lattice shape parameter $\tau$, and the average magnetic flux $b$ near $h_{c2}$,
via a Lyapunov-Schmidt reduction.


\subsection{Comments on the proof of 
Theorem \ref{thm:ALener}} \label{sec:approach-energy}

The relation between the Abrikosov function and the  average energy, $E_b(\tau) := \frac{1}{ |\Omega^\tau|}\E_{\Omega^\tau}(u_\om)$, of this solution is given by
\begin{proposition}
 In the case $\kappa > \frac{1}{\sqrt{2}}$, the minimizers,  $\tau_b$, of $\tau \mapsto E_b(\tau)$ are related to the minimizer,  $\tau_*$, of $\beta(\tau)$, as $\tau_b - \tau_* =O(\mu^{1/2})$, In particular, $\tau_b \to \tau_*$ as $b \to \kappa^2$.
\end{proposition}
This result was already found (non-rigorously) by Abrikosov~\cite{Abr}. Thus the problem of minimization of the energy per the lattice cell is reduced to finding the minima of  $\beta(\tau)$ as a function of the lattice shape parameter $\tau$.

Using 
symmetries of $\beta(\tau)$ one can also show (see 
\cite{ST3} and below) that  $\beta( \tau)$ has critical points at the points $\tau=e^{\pi i/3}$ and $\tau=e^{\pi i/2}$. However, to determine minimizers of $\beta( \tau)$ requires a rather delicate analysis, which 
gives
\begin{theorem}[\cite{ABN, NV}]
    The function $\beta(\tau)$ has exactly two critical points, $\tau
    = e^{i\pi/3}$ and $\tau = e^{i\pi/2}$. The first is minimum,
while the second is a maximum.
\end{theorem}
 Hence the second part of Theorem~\ref{thm:ALener} follows.

\subsection{Comments on the proof of Theorem \ref{thm:ALexist2}} \label{sec:approach-exist2}
The idea here is to reduce solving \eqref{GLE} for  $(\Psi, A)$  on the space $\R^2$ to solving it for $(\psi, a)$ on the fundamental cell $\Omega$, satisfying the boundary conditions 
\begin{equation} \label{gaugeperbc}
	\begin{cases}
	\psi(x + s) = e^{ig_s(x)}\psi(x), \\
	a(x + s) = a(x) + \nabla g_s(x), \\
	(\nu\cdot\nabla_a\psi)(x + s) = e^{ig_s(x)}(\nu\cdot\nabla_a\psi)(x), \\
\curl	a(x + s) = \curl a(x) ,\\
x\in  \p_1\Omega/\p_2\Omega\ \mbox{and}\  s=\w_1/\w_2.
		\end{cases}
\end{equation}
 induced by the periodicity condition \eqref{gauge-per}.
Here $\p_1\Omega/\p_2\Omega=$ the left/bottom boundary of $\Omega$, $  \{\w_1, \w_2\}$ is a basis in $\cL$ and $\nu(x)$ is the normal to the boundary at $x$.

To this end we show that, given a continuously differentiable function $(\psi, a)$ on the fundamental cell $\Omega$, satisfying the boundary conditions \eqref{gaugeperbc}, with $g_s$ satisfying  \eqref{cocycle-cond},
 we can lift it to a 
 continuous and continuously differentiable function $(\Psi, A)$  on the space $\R^2$, satisfying the gauge-periodicity conditions \eqref{gauge-per}. Indeed, we define
for any $\alpha \in \cL$,
\be  \la{lifting}
  \Psi (x)= \psi(x-\alpha)e^{i\Phi_\alpha (x)},\   A (x)= a(x-\alpha)+\nabla\Phi_\alpha (x),\   x\in\Omega+\alpha,
\end{equation}
where $\Phi_\alpha (x)$ is a real, possibly multi-valued, function to be determined. (Of course, we can add to it any $\cL-$periodic function.)  We define
\be  \la{Phial}  \Phi_\alpha (x):=   g_{\al}(x-\al),\  \mbox{for}\  x\in \Omega+\alpha.\end{equation}
\DETAILS{Finally, note that
 \begin{itemize}
 \item[(a)] Since  $\Psi, A$ satisfy the gauge-periodicity conditions \eqref{gaugeper}  in the entire space $\R^2$ and are smooth in $\R^2/(\cup_{s\in\cL}\p\Om)$, $\nabla_A\Psi$, $\Delta_A\Psi$ and $\curl^2 A$  are continuous and satisfy  the gauge-periodicity condition \eqref{gaugeper};
\item[(b)]  Since $u \equiv (\psi, a)$ satisfies the Ginzburg-Landau equations \eqref{gle} in $\Om$, then $U \equiv (\Psi, A)$  satisfies \eqref{gle} in $\R^2/(\cup_{t\in\cL}S_t\p\Om)$, where $S_t: x\ra x+t$;
\item[(c)]   Since  $\Psi, A$ satisfy the gauge-periodicity conditions \eqref{gaugeper}   in the entire space $\R^2$, we conclude by the first equation in
 \eqref{gle} that  $\Delta_A\Psi$ is continuous and satisfies the periodicity conditions (in the first equation of) \eqref{gaugeper} in $\R^2$ and therefore, by the Sobolev embedding, theorem so is $\nabla_A\Psi$. Hence, by the second equation in \eqref{GLE}, $\curl^2 A$ is continuous and satisfies the periodicity conditions \eqref{gaugeper} in $\R^2$.  Therefore, by iteration of the above argument (i.e. elliptic regularity),  $\Psi, A$ are smooth functions obeying \eqref{gaugeperbc} and  \eqref{GLE}.
 \end{itemize}}

 \begin{lemma} \la{lem:lifting}
Assume functions $(\psi, a)$ on $\Om$ are twice differentiable, up to the boundary, and obey the boundary conditions \eqref{gaugeperbc} and the Ginzburg-Landau equations \eqref{GLE}. Then the functions $ (\Psi, A)$, constructed in \eqref{lifting} - \eqref{Phial}, are smooth in $\R^2$ and satisfy the periodicity conditions \eqref{gauge-per} and the Ginzburg-Landau equations \eqref{GLE}.
\end{lemma}
\begin{proof} If $(\psi, a)$ satisfies the Ginzburg-Landau equations \eqref{GLE} in $\Om$, then $U \equiv (\Psi, A)$, constructed in \eqref{lifting} - \eqref{Phial},  has the following properties
 \begin{itemize}
 \item[(a)]    $(\Psi, A)$  is twice differentiable   
      and  
      satisfies \eqref{GLE} in $\R^2/(\cup_{t\in\cL}S_t\p\Om)$, where $S_t: x\ra x+t$;
\item[(b)]  $(\Psi, A)$  is  continuous  with continuous derivatives ($\nabla_A\Psi$ and $\curl A$) in $\R^2$ and  satisfies the gauge-periodicity conditions \eqref{gauge-per} in 
      $\R^2$.
 \end{itemize}
Indeed, the periodicity condition \eqref{gauge-per}, applied to the cells $\Omega+\alpha-\w_i$ and $\Omega+\alpha$ and the continuity condition on the common boundary of the cells $\Omega+\alpha-\w_i$ and $\Omega+\alpha$ imply that $\Phi_\alpha (x)$ should satisfy the following two conditions:
\be  \la{Phicond1}
  \Phi_\alpha (x)= \Phi_{\alpha-\w_i} (x-\w_i) +g_{\w_i}(x-\w_i),\ \mbox{mod}\ 2\pi,\ x\in \Omega+\alpha,
\end{equation}
\be  \la{Phicond2}
  \Phi_\alpha (x)= \Phi_{\alpha-\w_i} (x) +g_{\w_i}(x-\alpha),\ \mbox{mod}\ 2\pi,\ x\in \p_i\Omega+\alpha,
\end{equation}
where $i=1, 2,$ and, recall, $\{\w_1, \w_2\}$ is a basis in $\cL$ and $\p_1\Omega/\p_2\Omega$ is the left/bottom boundary of $\Omega$. 

To show that \eqref{Phial} satisfies the conditions \eqref{Phicond1} and \eqref{Phicond2}, we note that, due to \eqref{cocycle-cond}, we have $g_{\al}(x-\al)= g_{\al-\w_i}(x-\al) +g_{\w_i}(x-\w_i),\ \mbox{mod}\ 2\pi,\ x\in \Omega+\alpha,$ and $g_{\al}(x-\al)= g_{\al-\w_i}(x-\al+\w_i) +g_{\w_i}(x-\alpha),\ \mbox{mod}\ 2\pi,\ x\in \p_i\Omega+\alpha$, which are equivalent to \eqref{Phicond1} and \eqref{Phicond2}, with \eqref{Phial}.

The second pair of conditions in  \eqref{gaugeperbc} implies that $\nabla_A\Psi$ and $\curl A$ are continuous across the cell boundaries.

By  (a) and (b), the derivatives $\Delta_A\Psi$ and $\curl^2 A$ 
are continuous, up to the boundary,  in $S_t\p\Om,$ for every $ t\in\cL$. 
By \eqref{GLE}, they 
are equal in $\R^2/(\cup_{t\in\cL}S_t\p\Om)$ to functions continuous  in $\R^2$ satisfying there the periodicity condition \eqref{gauge-per}. 
Hence, they are also continuous and satisfy the periodicity condition  \eqref{gauge-per} in $\R^2$. By iteration of the above argument 
(i.e. elliptic regularity),  $\Psi, A$ are smooth functions obeying \eqref{gauge-per} and  \eqref{GLE}.
 \end{proof}
 
 Now, we use the $n-$vortex $(\Psi^{(n)}, A^{(n)})$, placed in the centre of the fundamental cell $\Omega$, to construct an approximate solution $(\psi^{\rm appr}, a^{\rm appr})$ to \eqref{GLE} in $\Omega$, satisfying \eqref{gaugeperbc}, and use it and the Lyapunov-Schmidt splitting technique to show that there is a true solution  $(\psi, a)$  nearby sharing the same properties. After that, we use Lemma \ref{lem:lifting} above to  lift $(\psi, a)$ to a solution $(\Psi, A)$  on the space $\R^2$, satisfying the gauge-periodicity conditions \eqref{gauge-per}. 

\subsection{Stability of Abrikosov lattices} \label{sec:ABR}
The Abrikosov lattices are static solutions to \eqref{GES} and 
their stability w.r. to the dynamics induced by these equations is an important issue. In \cite{ST2}, 
 we considered the stability of the Abrikosov lattices for magnetic fields  close to the second critical magnetic field $h_{c2}=\kappa^2$, under the simplest perturbations, namely those having the same (gauge-) periodicity as the underlying Abrikosov lattices (we call such perturbations \emph{gauge-periodic}) and proved 
for a lattice of arbitrary shape,  $\tau\in \C$, $\Im\tau > 0$, that, under gauge-periodic perturbations, Abrikosov vortex lattice solutions are\begin{itemize}
\item[(i)] \emph{  asymptotically stable for}  $\kappa^2 > \kappa_c(\tau)$;
 \item[(ii)]   \emph{ unstable for} $\kappa^2 < \kappa_c(\tau)$.
\end{itemize}

 This result belies the common belief among physicists and mathematicians that  Abrikosov-type vortex lattice solutions are stable only for triangular lattices and $\kappa> \frac{1}{\sqrt{2}}$, and  it seems this is the first time the threshold \eqref{kappac}
has been isolated.

In \cite{ST1}, similar results are shown to hold also for low magnetic fields close to $h_{c1}$.

Gauge-periodic perturbations are not a common type of perturbations occurring in superconductivity.
Now, we address the problem of the stability of Abrikosov lattices under local or  finite-energy perturbations (defined precisely below). 
We consider  Abrikosov lattices of arbitrary shape, 
  not just triangular or rectangular lattices as usually considered,  and for magnetic fields  close to the second critical magnetic field $h_{c2}=\kappa^2$. 

\paragraph{Finite-energy ($H^1-$) perturbations.}\label{sec:pert}

We now wish to study the stability of these Abrikosov lattice solutions under a class of perturbations that have finite-energy.
More precisely, we fix an Abrikosov lattice solution $u_\om$ and consider perturbations $v : \R^2 \to \C \times \R^2$ that satisfy
\begin{equation}\label{Lambda}
\Lambda_{u_\om}(v) = \lim_{Q\ra\R^2} \big(\E_{Q}(u_\om + v) - \E_{Q}(u_\om)\big) < \infty.
\end{equation}
Clearly, $\Lambda_{u_\om}(v)<\infty$ for all vectors of the form $v=T^{\rm gauge}_\gamma  u_\om -  u_\om$, where $\g \in  H^2(\R^2;\R)$.

 In fact, we will be dealing with the smaller class,  $ H^1_{\textrm{cov}} $,  of perturbations, where $ H^1_{\textrm{cov}} $ 
 is the Sobolev space of order $1$ defined by the covariant derivatives, i.e., $$ H^1_{\textrm{cov}}:=\{v\in L^2(\R^2, \C\times \R^2)\ |\ \|v\|_{H^1}< \infty \},$$ where the norm $\|v\|_{H^1}$ is determined by
 the covariant inner product
\begin{align*}
    \langle v, v' \rangle_{H^1} = \Re \int \bar{\xi}\xi' + \overline{\COVGRAD{A_\om}\xi} \cdot \COVGRAD{A_\om}\xi' + \alpha \cdot \alpha' + \sum_{k=1}^2 \nabla\alpha_k\cdot\nabla\alpha'_k, 
\end{align*}
where $v=( \xi, \al),\ v'=( \xi', \al')$, 
 while  the $L^2-$norm  is given by 
\begin{equation} \label{L2-inner-product}
    \langle v, v' \rangle_{L^2} = \Re \int \bar{\xi}\xi' + \alpha \cdot \alpha'.
\end{equation}
An  explicit representation for the functional   $\Lambda_{u_\om}(v)$, given below shows (see \eqref{Lambda-expres}) that  $\Lambda_{u_\om}(v)<\infty$ for all vectors  $v\in  H^1_{\textrm{cov}} $. 

To introduce the notions of stability and instability, we note that the hessian Hess $\E(u)$ is 
well defined as a differential operator  for say $ u\in u_\om +  H^1_{\textrm{cov}} $ and is a real-linear operator on $H^1_{\textrm{cov}}$.
We define the manifold  $$\mathcal{M}_\om = \{ T^{\rm gauge}_\gamma  u_\om : \g \in H^1 (\R^2, \R) \}$$ of gauge equivalent Abrikosov lattices and the $H^1-$distance, $\dist_{H^1}$, to this manifold.

  \begin{definition}\label{def:stability} We say that the Abrikosov lattice $u_\om$ is {\it asymptotically stable} under $H^1_{\textrm{cov}} -$ perturbations, 
  if there is $\del>0$ s.t. for any initial condition $u_0$ satisfying $\dist_{H^1} (u_0, \mathcal{M}_\om)\le \del$ 
there exists $g(t)\in H^1$, s.t.   the solution $u(t)$ of \eqref{GES}  satisfies 
   $ \| u(t)-T^{\rm gauge}_{g(t)}  u_\om\|_{H^1} \ra 0$,    as $t \ra \infty$.  We say that $u_\om$ is {\it energetically} unstable if 
   the hessian, $\E''(u_\om)$, of $\E(u)$ at $u_\om$ has a negative spectrum. 
 \end{definition}

We restrict the initial conditions $(\Psi_0, A_0)$ for \eqref{GES} 
 satisfying 
\begin{equation} \label{parity0}
T^{\rm refl} (\Psi_0, A_0)=(\Psi_0, A_0).	 
\end{equation}
Note that, by  uniqueness,  the Abrikosov lattice solutions $u_\om = (\Psi_\om, A_\om)$ satisfy $T^{\rm refl} u_\om=u_\om$ and therefore so are the perturbations, $v_0:= u_0-u_\om$, where $u_0:= (\Psi_0, A_0)$:
\begin{equation} \label{parity}
	T^{\rm refl} v_0=v_0.	 
\end{equation}

\DETAILS{solutions $(\Psi, A, \Phi)$ of \eqref{GES} for which $ \Psi$ and $ \Phi$ are even and  $ A$ are odd under the reflections: 
\begin{equation} \label{parity''}
	\cR \Psi = \Psi,\ \cR \Phi = \Phi,\ \cR A = -A. 
\end{equation}
{\bf Relevance of such solutions depends on the parity of $u_\om$, if it has one.}
 {\bf (Does $u_\om= (\Psi_\om, A_\om)$ have the parity? Is there a symmetry breaking here?)}}


\paragraph{Stability result.}\label{sec:main-res}

Recall that $\beta(\tau)$  is  the Abrikosov 'constant', introduced in \eqref{beta}. 
 \begin{theorem}\label{thm:stability} 
 There exists a modular function  $\gamma(\tau)$ depending on the lattice shape parameters $\tau$,
 such that, for $b$ sufficiently close to $\kappa^2$,  in the sense of  \eqref{b-cond}, 
and, under 
$H^1-$perturbations, satisfying \eqref{parity0},  
  \begin{itemize} \item the Abrikosov lattice $u_\om$ is asymptotically stable 
 for all $(\tau,  \kappa)$ s.t.  $\kappa >\frac{1}{\sqrt 2}$ and $\g (\tau) >0$   
 and 
 \item energetically unstable otherwise.
  \end{itemize} 
 \end{theorem}
%
The function  $\g( \tau),\ \im\tau>0,$ appearing in the theorem above is described below. Meantime we make the following important remark.
Since we know that, for $\kappa>1/\sqrt 2$, the triangular lattice has the lowest energy (see Theorem \ref{thm:ALexist1}), this seems to suggest that other lattices should be unstable. 
  The reason that this energetics does not affect the stability under local perturbations can be gleaned from investigating the zero mode of the Hessian of the energy functional associated with different lattice shapes, $\tau$. This mode is obtained by differentiating the  Abrikosov lattice solutions w.r.to $\tau$, which shows that it grows linearly in $|x|$. To rearrange a non - triangular  Abrikosov lattice into  the triangular one, one would have to activate this mode and hence to apply a perturbation,  growing at infinity (at the same rate).

  This also explains why the Abrikosov 'constant'  $\beta(\tau)$ mentioned above, which plays a crucial role in understanding the energetics of the Abrikosov solutions, 
   is not directly related to the stability under local perturbations, the latter is governed by   $\gamma(\tau)$.

\subsection{The function  $\g( \tau)$}  \label{sec:gamma} 
  \begin{theorem}\label{thm:gammak}  The function  $\g( \tau)$ 
  on lattice shapes  $\tau$, entering Theorem \ref{thm:stability}, is given by 
\begin{equation}\label{gamma}
\g( \tau):=\inf_{\chi\in \hat\LAT_\tau}  \g_\chi( \tau),\ \text{  where  }\ 
	\g_\chi( \tau) := 2\lan|\phi_0|^2|\phi_\chi|^2\ran_{\cT_\tau} + |\lan \phi_0^2\bar{\phi}_\chi \bar{\phi}_{\chi^{-1}}\ran_{\cT_\tau} | - \lan |\phi_0|^4 \ran_{\cT_\tau} .\end{equation}
 Here the functions  $\phi_\chi,\ \chi\in \hat\LAT_\tau,$  are unique solutions of the equations
\begin{align}\label{phichi-eqs} (-\COVLAP{A^1}-1)\phi =0,\ \quad \phi (x+s) =e^{i\frac12 s\cdot J x} \chi(s)
\phi (x),\ \quad \forall s\in \LAT_\tau,
	\end{align}
 with $A^1(x):= - \frac12  J x$ and $J:= \left( \begin{array}{cc} 0 & 1 \\ -1 & 0 \end{array} \right)$,  normalized as $\lan|\phi_\chi|^2\ran_{\cT_\tau} =1$.  
\end{theorem}
  For the function  $\g( \tau),\ \im\tau>0,$ defined in  \eqref{gamma}, 
 has the following properties 
  \begin{itemize} 
 \DETAILS{ \item $\g( \tau),\ \im\tau>0,$  
a modular function, symmetric w.r.to the imaginary axis, specifically, 
 \begin{align}\label{gam-propert} &\g( \tau+1)= \g( \tau),\  \g( -\tau^{-1})= \g( \tau),\ \g (- \bar\tau)= \g( \tau); \end{align}  
 \item  
$\g( \tau),\ \im\tau>0,$  independent of the choice of the dual cell $ \Om_\tau^*$;  }
\item $\g( \tau),\ \im\tau>0,$   is symmetric w.r.to the imaginary axis,  
 $\g (- \bar\tau)= \g( \tau);$

\item  $\g( \tau)$ has critical points at   $\tau=e^{i\pi/2}$ and $\tau=e^{i\pi/3}$, provided it is differentiable at these points.
  \end{itemize}  
We see also that 
the  the Abrikosov constant, $\beta(\tau)$, is related to $\g_k( \tau)$   as  $\beta(\tau) 
=\frac12 \g_0( \tau)$.

  The  function $\g (\tau)$ is studied numerically in \cite{ST3}, where the above conjecture is confirmed and 
   is shown that it 
becomes  negative for $\im\tau\ge 1.81$. Moreover, it is computed that
  \begin{align}\label{gam-value} |\g ( \tau) - c| &\le   7.5 \cdot 10^{-3},\ \ \quad \mbox{where}\  \notag\\ & 
c =  0.64\ \mbox{for}\  \tau=e^{i\pi/3}  \quad \mbox{and}\   \quad c= 0.4\ \mbox{for}\  \tau=e^{i\pi/2}.  \end{align}

The definition of $\g( \tau)$ implies that it is  symmetric w.r.to the imaginary axis, $\g (- \bar\tau)= \g( \tau)$.  Hence it suffices to consider $ \g( \tau)$ on the $\Re\tau \geq 0$ half of  the fundamental domain, $\Pi^+/SL(2, \Z)$, 
 of  the modular group  $SL(2, \Z)$ (the heavily shaded area on the Fig. \ref{fig:PoincareStrip} above).
 \DETAILS{
\begin{figure}[h!]
	\centering 

  \includegraphics[width=2.5in]{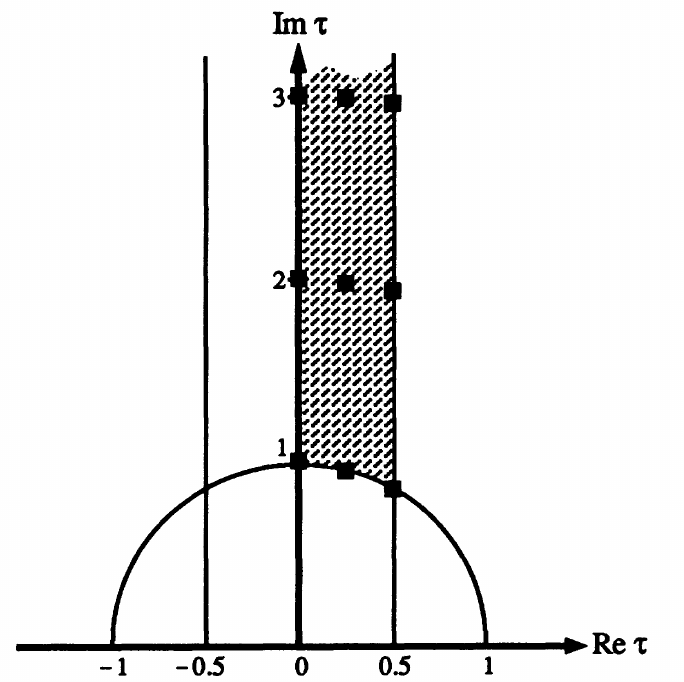} 
  \caption{Fundamental domain  of $\gamma(\tau)$. }\label{fig:PoincareStrip}
\end{figure}  }
Using the symmetries of $\g( \tau)$, we show in \cite{ST3} that
 the points  $\tau=e^{i\pi/2}$ and $\tau=e^{i\pi/3}$ are critical points of the function $\g( \tau)$, provided it is differentiable at these points. (While the functions  $\g_k( \tau)$ are obviously smooth, derivatives of $\g( \tau)$ might jump. In fact, the numerical computations described in   \cite{ST3} show that  $\g( \tau)$ is likely to have the line of cusps at $\re \tau =0$.)  We {\it conjecture}: 
\begin{itemize}


\item 
For fixed $\re \tau\in [0, 1/2]$, $\g( \tau)$ is a decreasing function of  $\im \tau$. 
\item 
$\g( \tau)$  has  a unique global maximum  at $\tau=e^{\frac{i\pi}{3}}$  and a saddle point at $\tau=e^{\frac{i\pi}{2}}$.
\end{itemize}
 In \cite{ST3}, we confirm this conjecture numerically (numerics is due to Dan Ginsberg, see Figure \ref{fig:plot}  for the result of 
computing $\g( \tau)$  in Matlab, using the default Nelder-Mead algorithm).

 \bigskip

 \begin{figure}[h!]
	\centering 
	{\includegraphics[width=3.1in]{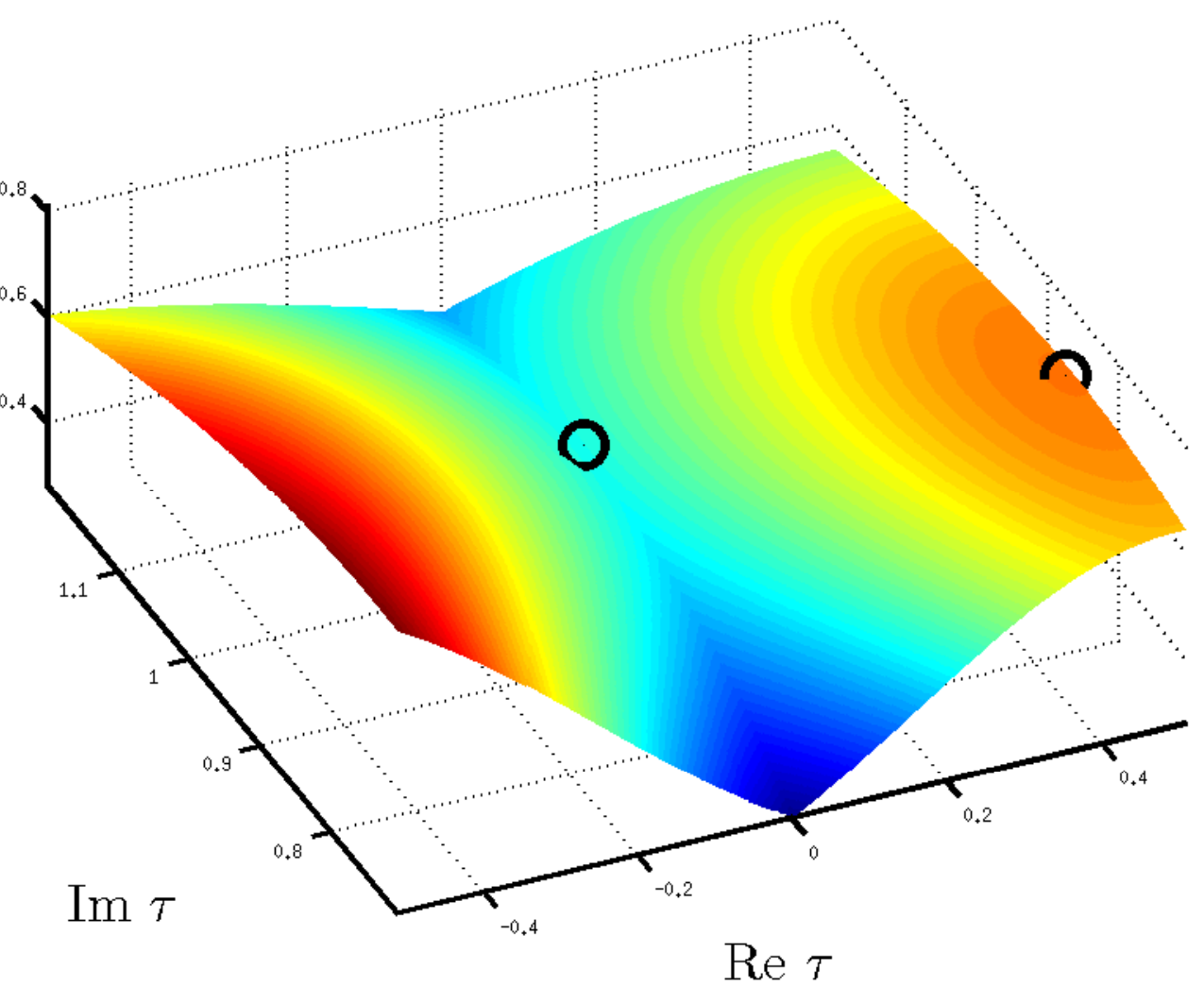}}
	{\includegraphics[width=3.1in]{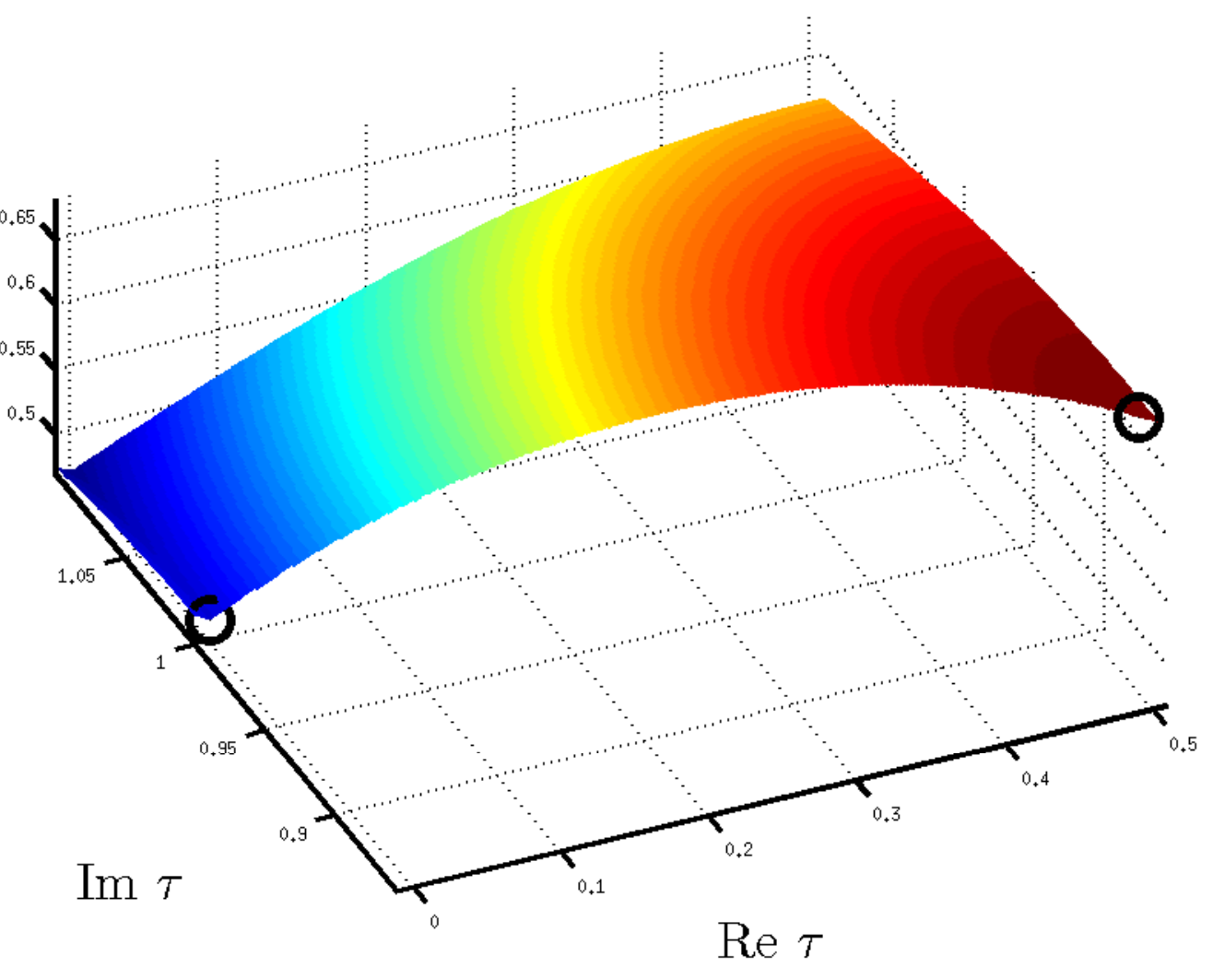}}
	\caption{Plots of the function $\gamma^{approx}(\tau)$. Computed in Matlab
	on a uniform grid with step size 0.01. The plot on the right is the function 
	plotted only on the Poincar\'{e} strip. The circled points are $\tau = e^{i\pi/2}$
	and $\tau = e^{i\pi/3}$.}
	\label{fig:plot}
\end{figure}
%
\DETAILS{ \begin{figure}[h!]
	\centering 
	{\includegraphics[width=3.1in]{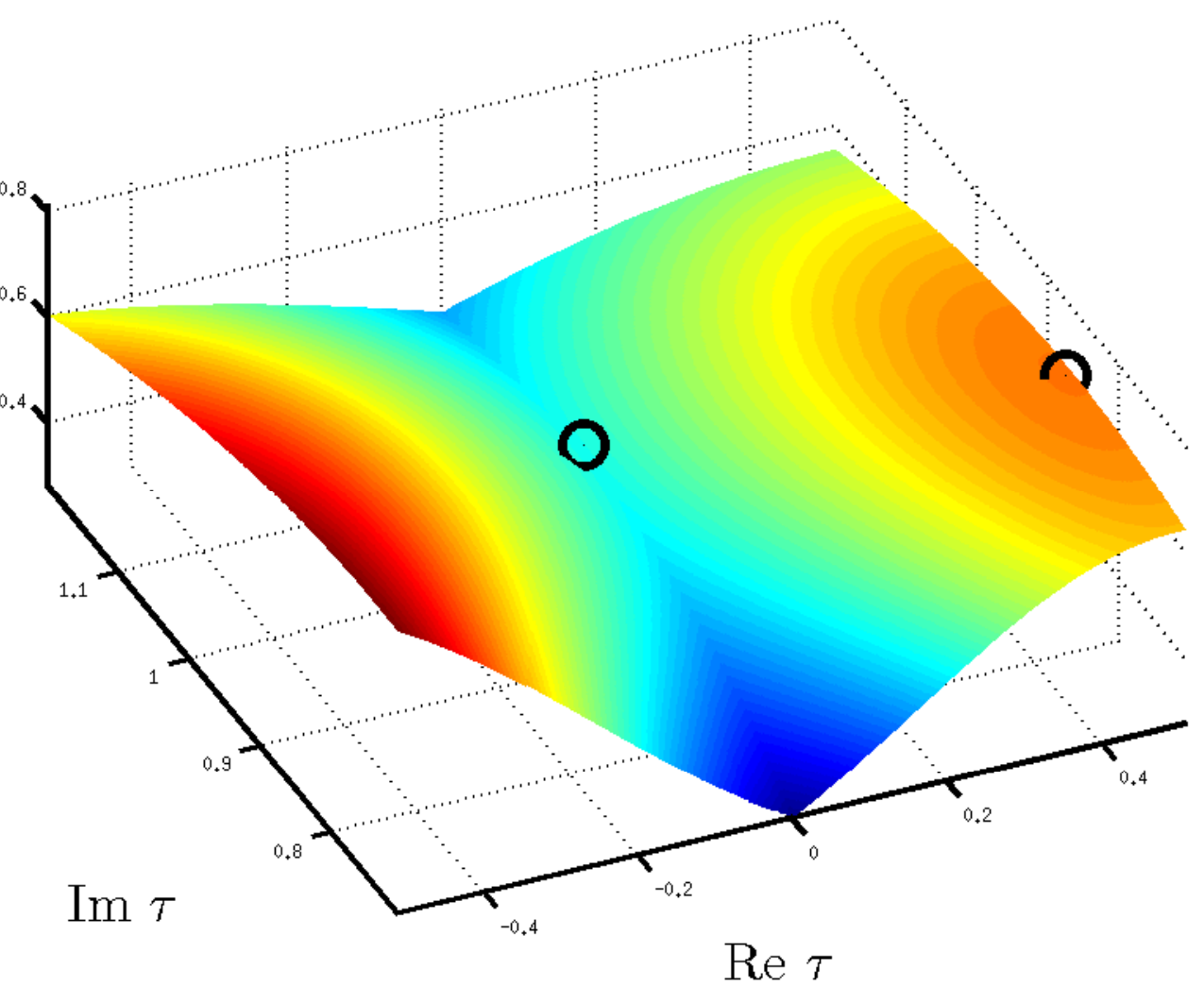}}
	{\includegraphics[width=3.1in]{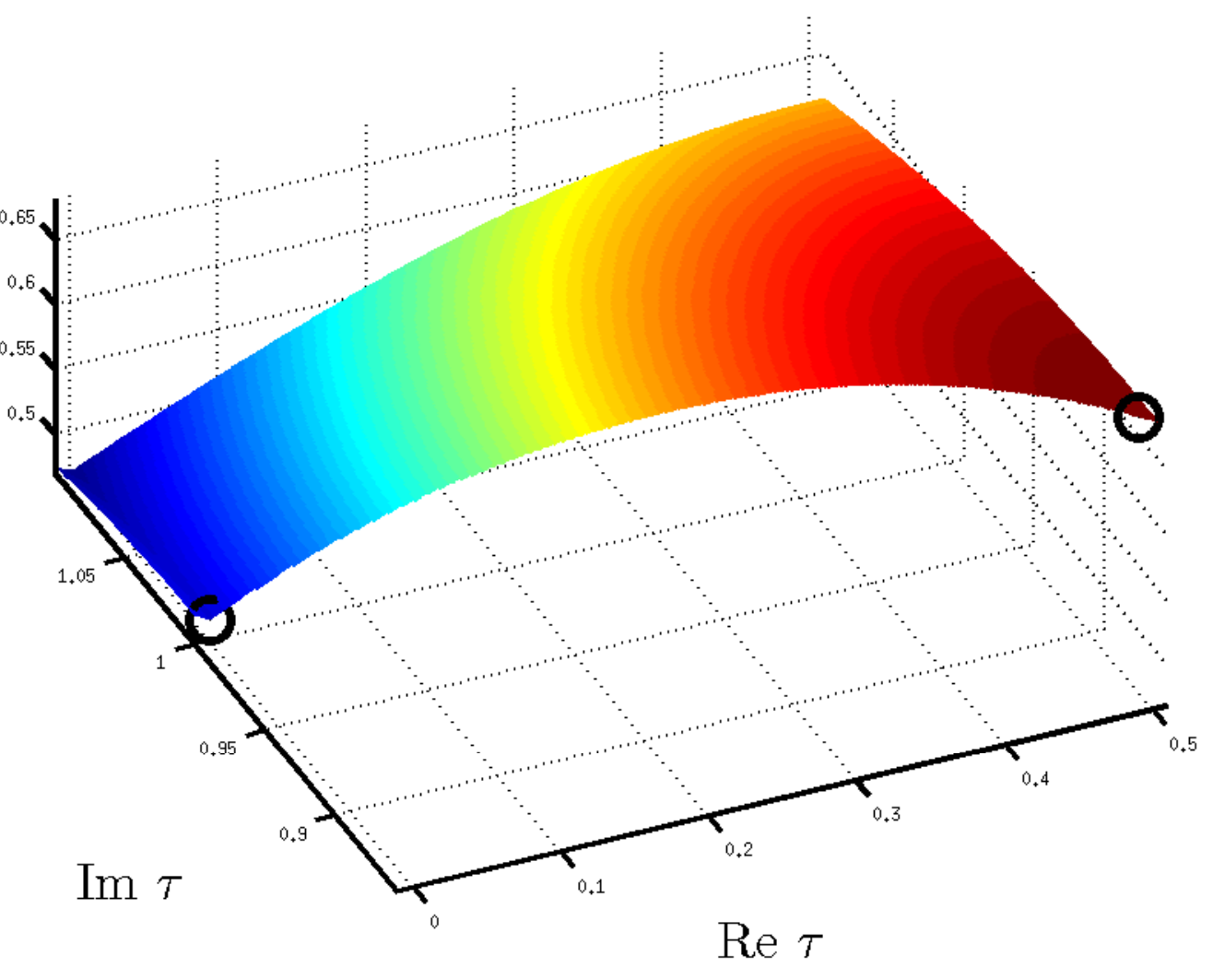}}
	\caption{Plots of the function $\gamma^{approx}(\tau)$. Computed in Matlab
	on a uniform grid with step size 0.01. The plot on the right is the function 
	plotted only on the Poincar\'{e} strip. The circled points are $\tau = e^{i\pi/2}$
	and $\tau = e^{i\pi/3}$.}
	\label{fig:plot}
\end{figure}


\begin{figure}[h!]
	\centering 
	\includegraphics[width=3.1in]{contour.pdf}
	\caption{Plot of the gradient of $\gamma^{appr}_k(\tau)$, computed on the same
	grid as in the previous figure. The two marks indicate the points $\tau = e^{i\pi/2}$
	and $\tau = e^{i\pi/3}$. }
	\label{fig:contour}
\end{figure}

\begin{figure}[h!]
	\centering 
	\includegraphics[width=5.1in]{zeroset.png}
	\caption{Plot of the zero set of $\gamma^{appr}(\tau)$. }
	\label{fig:contour}
\end{figure} }

Calculations of   \cite{ST3} show that  $\gamma(\tau)>0$  for all equilateral lattices, $|\tau|=1$, and  is negative for  $|\tau| \ge 1.3$. Though Abrikosov lattices are not as rigid under finite energy perturbations, as for gauge-periodic ones, they are still surprisingly stable.

  It is convenient to consider $\g ( \tau)$ as a modular function on  $\Pi^+$.  To this end, recalling 
  that we identify $\R^2$ with $\C$, via the map $(x_1, x_2)\ra x_1+i x_2$, 
we can choose a basis in $\LAT$  so that $\LAT=r  (\Z+\tau\Z)$, where $\tau\in \C$, $\Im \tau>0$, and $r>0$. By   the quantization condition \eqref{quant-cond} with $n=1$, $r:=\sqrt{\frac{2\pi }{\im\tau b} }$. 
Denote $\LAT_\tau=\sqrt{\frac{2\pi }{\im\tau } }  (\Z+\tau\Z)$. 
The dual to it is $\cL^*_\tau=\sqrt{\frac{2\pi}{\im\tau} } i (\Z-\tau\Z)$.  (The dual, or reciprocal, lattice, $\LAT^*$,  of $\LAT$ consists of all vectors $s^* \in \R^2$ such that $s^* \cdot s \in 2\pi\Z$,  for all $s \in \LAT$.)    
\DETAILS{Let  $\Om_\tau$ is a fundamental cell of the lattice  $\LAT_\tau=\sqrt{\frac{2\pi n}{\im\tau } }  (\Z+\tau\Z)$  and  let  $\Om^*_\tau$ 
be  a fundamental cell of $\cL^*_\tau$, 
  chosen so that $\Omega^*_\tau$ is invariant under reflections, $k \ra -k$.}  
  We identify the dual group, $\hat \LAT_\tau$, with a fundamental cell, $\Om^*_\tau$, of the dual lattice  $\cL^*_\tau$,  chosen so that $\Omega^*_\tau$ is invariant under reflections, $k \ra -k$.  This  identification given explicitly by 
 $\chi(s)\ra   \chi_k(s) = e^{ik\cdot s}\leftrightarrow k$. 
 
 Then the functions  $\g( \tau)$ and $\g_k( \tau) =\g_{\chi_k}(  \tau) ,\ 
 \chi_k(s) = e^{ik\cdot s},\   k\in \Omega^*_\tau$, are defined  in \eqref{gamma} for all  $\tau, \Im\tau>0$. Since they are independent of the choice of a basis in $\LAT_\tau$, they are invariant under action  of  the modular group  $SL(2, \Z)$, $\g( g\tau)= \g( \tau),\ \forall g\in SL(2, \Z)$, i.e. they  modular functions on  $\Pi^+$.   
%

The numerics mentioned above are based on the following  explicit representation of  the functions    $\g_k( \tau)$: 
    \begin{theorem}\label{thm:gamma} The functions    $\g_k( \tau)$ admit the explicit representation 
  \begin{align}\label{gamk-series}
\g_k( \tau)& =2\sum_{t\in \LAT_{\tau}^*} e^{- 
\frac12  |t|^2  } \cos [  \im (  \bar k t)]  + |\sum_{t\in \LAT_{\tau}^*} e^{- 
\frac12   |t+k| ^2     +   i  \im (  \bar k t) }|- \sum_{t\in \LAT_{\tau}^*} e^{- 
\frac12 |t|^2  } .
\end{align}
\end{theorem}
Our computations show also that  
 \begin{itemize}
\item $ \g_k( \tau)$ is minimized at  $k \approx \sqrt{\frac{2\pi }{\im\tau b} } (\frac{1}{2} - \frac{1}{2\sqrt{3}} i)$ at the point $\tau = e^{i \pi /3}$, and a value of
 $k \approx \sqrt{\frac{2\pi }{\im\tau b} } (\frac{1}{2} + i\frac{1}{2})$ for $\tau = e^{i \pi/2}$, which corresponds to vertices of the corresponding Wigner-Seitz cells. 
\end{itemize}
Interestingly,  in \cite{ST3}, we show that  the points  $k\in \frac12\cL^*_\tau$ 
  are critical points of the function $\g_{k}( \tau)$ in $k$.
  \DETAILS{ and we  {\it conjecture} that 
\begin{itemize}

\item   The points  $k\in\frac12\cL^*_\tau$ 
are the only critical points of the function $\g_{k}( \tau)$ in $k$.  
\end{itemize}
If the latter conjecture is true, this would considerably simplify proving the first conjecture.}  It is easy to see that $k = 0$ is a point of maximum of   $\g_k( \tau)$ in  $k\in \Om^*_\tau$.

\medskip


\noindent
{\bf Remark}.
  We think of  $\g_{k}(\tau) $ as the 'Abrikosov beta function with characteristic' (while $\beta(\tau) $ is defined in terms of the standard theta function,  $\g_{k}(\tau) $ is defined in terms of theta functions with finite characteristics,  see below).

\DETAILS{2) As for the function $\beta(\tau)$, 
the property \eqref{gam-propert} has the origin in the fact that the function $\g(\tau)$ which appears in the eigenvalue asymptotics does not depend on a choice of a basis in $\LAT_\tau$. }
\DETAILS{Now, the parameter $\tau$ in the formulae above determines the basis   $\sqrt{\frac{2\pi}{\im\tau} }(1 ,\  \tau)$  of the lattice $\LAT_\tau$ and any two bases in $\LAT_\tau$  are related by a modular map  $\sqrt{\frac{2\pi}{\im\tau} } (\nu_1,\ \nu_2)\ra  \sqrt{\frac{2\pi}{\im\tau} }(\nu_1',\ \nu_2')=\sqrt{\frac{2\pi}{\im\tau} }(\al\nu_1+\beta,\ \g\nu_2 +\del)$, where  $\al, \beta, \g, \del\in \Z$,  and 
$\al \del-\beta \g=1$ (with the matrix  $\left( \begin{array}{cc} \al & \beta \\ \g & \del \end{array} \right)$, an element  of the modular group  $SL(2, \Z)$). In particular, if we replace the basis  $\sqrt{\frac{2\pi}{\im\tau} }(1, \tau)$  by the basis $\sqrt{\frac{2\pi}{\im\tau} }( \g \tau+\del,\  \al \tau+\beta )$, our results should not change. Under this map, the shape parameter is mapped as $ \tau\ra g\tau:=\frac{\al\tau+\beta}{\g\tau+\del}$.}   



\subsection{The key ideas of approach} \label{sec:approach}
    Let   $u_\om = (\Psi_\om, A_\om),\ \om:=  (\tau, b, 1),$  be  an Abrikosov lattice solution. As usual, one begins with the  Hessian $L_\om:=\E''(u_\om)$ of the energy functional $\E$ at $u_\om$.  To begin with, due to the fact that  the solution $u_\om$ of \eqref{GLE} breaks the gauge invariance, the operator $L_\om$ has the  gauge zero modes,  $L_\om  G_{ \g'} =0$, where  $ G_{\gamma'} := (i\gamma'\Psi_\om, \nabla\gamma')$.  Then the stability of the static solution $u_\om$ is decided by the sign of the 
infimum $\mu (\om, \kappa):= \inf_{v \in H^1_\perp}   \lan v, L_\om v \ran_{L^2}/ \|v\|_{L^2}^2$,  on the subspace $H^1_\perp$ of the Sobolev space $H^1_{\textrm{cov}} $, which is the orthogonal complement of the these zero modes $ G_{\gamma'}$.  

The key idea of the proof of the first part of Theorem \ref{thm:stability} stems from the observation that since the  Abrikosov lattice solution $u_\om = (\Psi_\om, A_\om)$ is gauge periodic (or equivariant) w.r.to  the lattice $\LAT_\om$, 
i.e. satisfies \eqref{gauge-per'} - \eqref{cocycle-cond},
 the linearized map $L_\om$ commutes with magnetic translations,  
 \begin{align}\label{rhot'}\rho_s = T_s^{\rm mag-trans}  \oplus T^{\rm trans}_s,\ 
  \quad  \forall s\in \cL_\om,\end{align}
 where  $T_s^{\rm mag-trans} = 
(T^{gauge}_{g_s})^{-1} T^{\rm trans}_s$ is the group of magnetic translations and, recall,  $T^{\rm trans}_s$ denotes translation by $s$, which, due to \eqref{cocycle-cond}, give a unitary group representation of $\LAT_\om $. (Note that \eqref{gauge-per'} implies that $u_\om$ is invariant under the magnetic translations, $T_s^{\rm mag-trans} u_\om = u_\om$.) 
 Therefore $L_\om$ is unitary equivalent to a fiber integral over the dual group, $\hat\LAT_\om$, of the group of lattice translations, 
\begin{align}\label{L-deco} L_\om \approx \int_{\hat\LAT_\om}^\oplus L_{\om \chi} \hat{d\chi}\ \quad \mbox{acting on}\ \quad  \int_{\hat\LAT_\om}^\oplus \cH_{ \chi} \hat{d\chi}.\end{align}
Here $\hat{d\chi}$ is the usual Lebesgue measure normalized so that $\int_{\hat\LAT_\om} \hat{d\chi} = 1$, $L_{\om \chi}$ is the restriction of $L_{\om}$ to 
$\cH_{\chi }$ and $\cH_{\chi }$ is the set of all functions, $v_k$, from $L^2_{\rm loc}(\R^2;  \C\times \R^2)$, 
satisfying the gauge-periodicity conditions
	\begin{equation}\label{bc'}\rho_s v_\chi (x) = \chi (s) v_\chi (x),\ \forall s\in\ 
	\LAT_\om,\end{equation}
 where 
  $\chi : \LAT_\om \to U(1)$ are   the characters acting on $v = (\xi, \alpha)$ as the multiplication operators
\begin{align*}\chi (s)v = 
(\chi (s)\xi, \chi (s)\alpha).\end{align*}
 Furthermore, $\cH_{\chi }$  is endowed with the  inner product $  \lan v, v' \ran_{L^2} = \frac{1}{|\cT_\om|}\int_{\cT_\om} \re\bar{\xi}\xi' + \bar{\alpha}\alpha' ,$ where  $\cT_\om:= \R^2/\LAT_\om$ is a $2-$torus 
 and $v = (\xi,  \alpha),\  v' = (\xi', \alpha')$. 
 	The  inner product in $\int_{\hat \LAT_\om}^\oplus \cH_{\chi } \hat{d\chi }$ is given by $\lan v , w \ran_{\cH }:=\frac{1}{|\hat\LAT_\om|} \int_{\hat\LAT_\om}\lan v_\chi , w_\chi \ran_{\cH_\chi }  \hat{d\chi }$.

 The decomposition \eqref{L-deco} implies that   the smallest spectral point, $\mu_\om ( \kappa)$, of $L_\om$ is given by 
 $\mu_\om ( \kappa)=\inf_{k\in \Omega^*_\om} \mu_{\om, \chi }(\kappa)$, where $\mu_{\om, \chi }(\kappa)$ are the smallest eigenvalues of $L_{\om \chi }$. The spectral analysis of fibers $L_{\om \chi },\ \chi \in \hat\LAT_\om^*,$  using a natural perturbation parameter $\e$  defined as
\begin{equation}\label{eps} 
	\e = \sqrt{\frac{\kappa^2 - b}{\kappa^2[(2\kappa^2 - 1)\beta(\tau) + 1]}},
\end{equation} gives 
leads to the following expression for $\mu_{\om, k}(\kappa)$ 
The lowest eigenvalue, $\mu_{\om, \chi }(\kappa)$, of the operator $L_{\chi } $ 
	on the subspace $\cH_{\chi \perp}$ of $\cH_{\chi }$, which is the orthogonal complement of the the $\chi -$fibers of the gauge zero modes $ G_{\gamma'} := (i\gamma'\Psi_\om, \nabla\gamma')$, are of the form
		\begin{equation}\label{muk}
		\mu_{\om, \chi }(\kappa) = 
		 (\kappa^2  - \frac{1}{2}) \g_\chi ( \tau) \e^2 + O(\e^3), \end{equation}
	for $\chi \ne 1$, 	 with   the functions $\g_\chi ( \tau)$ defined in \eqref{gamma} 
	and  \eqref{phichi-eqs}. 
	
The linear result above gives the linearized (energetic) stability	of $u_\om$, if $\mu_{\om, \chi }(\kappa)>0$, and the instability, if $\mu_{\om, \chi }(\kappa) < 0$. To lift the  stability part to the (nonlinear) asymptotic stability, we use the functional   $\Lambda_{u_\om}(v)$, given in \eqref{Lambda}. It has the following explicit representation
	\begin{equation}	\label{Lambda-expres}
	\Lambda_\om(v) = \frac{1}{2}\lan v, L_\om v \ran_{L^2} + R_\om(v),
	\end{equation}
	where  $L_\om:=\E''(u_\om)$ is  the  Hessian of the energy functional $\E$ at $u_\om$, $\lan v, v' \ran_{L^2}$ is the $L^2$ inner product, \eqref{L2-inner-product}, and  $R_\om(v)$ is given  by
	\begin{align}	\label{remainder}
	R_\om(v) &= \int (|\alpha|^2 + \kappa^2|\xi|^2)\Re(\bar{\Psi}_\om\xi) \notag \\ &- \alpha \cdot \Im(\bar{\xi}\COVGRAD{A_\om}\xi) + \frac{1}{2}(|\alpha|^2 + \frac{\kappa^2}{2}|\xi|^2) |\xi|^2.
	\end{align}
Using this expression, we obtain appropriate differential inequalities for  $\Lambda_{u_\om}(v)$, which imply the asymptotic stability.  $\Box$ 

\medskip


\noindent
{\bf Remark}.
$\mathscr{H}_\chi $ can be thought of as the space of  $L^2-$sections of the vector bundle $[\R^2\times  (\C\times \R^2)]/\cL$, with the group $\cL$, which acts on $\R^2\times (\C\times \R^2)$ as $s (x, v)= (x + s,  \chi_k (s) \tau_{g_s(x)}^{-1}  v)$, where $\tau_{\al} = e^{i\al} \oplus e^{-i\al} \oplus \one \oplus \one$, for $\al\in \R$.

Before proceeding, we recall that for  $\LAT $, considered as a group of lattice translation, the dual group $\hat \LAT$  
 is the group of all continuous homomorphisms from  $\LAT $ to $U(1)$, i.e. the group of characters, $\chi_k(s) : \LAT \to U(1)$). ($\hat \LAT$ can be identified with the fundamental cell  $\Om^*$ of the dual lattice  $\LAT^* $, with the  identification given explicitly by identifying $k \in \Omega^* $ with the character $\chi_k : \LAT \to U(1)$ given by
 $   \chi_k(s) = e^{ik\cdot s}.$

Now, we derive the explicit representation \eqref{gamk-series} and the uniqueness for \eqref{phichi-eqs}. 
 As in \eqref{theta-phi}, we introduce the new function $\theta_{q}(z, \tau)$, by the equation 
\begin{align}\label{phik-thetaq}
\phi_k(x)= c_0 
e^{\frac{\pi}{2 \im\tau}(z^2 -|z|^2)} \theta_{q}(z, \tau), \end{align} 
 where 
 $c_0$ is such that  	 $\lan |\phi_k|^2 \ran_{\Om_\tau} = 1$,  and 
   $x_1+i x_2=\sqrt{\frac{2\pi}{\im\tau} }    z \quad \mbox{and}  \quad k=\sqrt{\frac{2\pi}{\im\tau} }   i q.$ 
Then, again as above, we can show that  the functions $\theta_{q}(z, \tau)= e^{-\frac{\pi}{2 \im\tau}(z^2 -|z|^2)} \phi_k(x)$ 
are entire functions  (i.e. they solve $
	\bar\p \theta_q = 0$) and  satisfy the periodicity conditions
	\begin{align}\label{thetaq-per1}	&\theta_q(z + 1, \tau) = e^{-2\pi i a  } \theta_q(z, \tau), \\
	&\theta_q(z + \tau, \tau) = e^{- 2\pi i b}e^{-i\pi\tau-2\pi i z} \theta_q(z, \tau), \label{thetaq-per2}	\end{align}
where $a, b$ are real numbers defined by $q=-a\tau +b$.
 This shows that $\theta_{q}(z, \tau)$ are the theta functions with characteristics $q$, and this characteristics is determined by $k$, which  in physics literature is called (Bloch) quasimomentum. Moreover, it has the following series expansion
	 \begin{align}\label{thetaq} \theta_{q}(z, \tau):= 
e^{\pi i (a^2\tau -2 ab-2 a  z)} \sum_{m=-\infty}^{\infty} e^{2\pi i  q m} e^{\pi i m^2\tau}  e^{2\pi i m z} .	\end{align}


\DETAILS{\section{Vortex dynamics}

One of the themes of recent mathematical studies of the
Ginzburg-Landau equations is an attempt to describe solutions
in terms of the basic localized structures -- namely vortices
-- yielding a sort of finite-dimensional reduction. In
particular, for various dynamical Ginzburg-Landau equations,
one wants to understand the structure and dynamics of the
vortices in solutions.}

\section{Multi-vortex dynamics}
\label{sec:vdyn}

Configurations containing several vortices are not, in
general, static solutions. Heuristically, this is due
to an effective inter-vortex interaction, which causes
the vortex centers to move. It is natural, then, to seek
an effective description of certain solutions of time-dependent
Ginzburg-Landau equations in terms of the vortex locations
and their dynamics -- a kind of finite-dimensional reduction.
In recent years, a number of works have addressed this problem,
from different angles, and in different settings. We will
first describe results along these lines from~\cite{GS2} for
magnetic vortices in $\R^2$, and then mention some other
approaches and results.


\paragraph{Multi-vortex configurations.}

Consider test functions describing several vortices, with the
centers at points $z_1$, $z_2, \ldots, z_m$, and with degrees
$n_1, n_2, \ldots, n_m$, ``glued together''. The simplest example
is
$  \vz = (\Psi_{\zb,\chi}, A_{\zb, \chi}),$ with
\[
  \Psi_{\zb,\chi}(x) = e^{i\chi(x)} \prod_{j=1}^m
  \Psi^{\nj}(x-z_j),\] 
\[  A_{\zb,\chi}(x) = \sum_{j=1}^m A^{\nj} (x-z_j) +
  \nabla \chi(x) \; ,
\]
where $\zb = (z_1, z_2, \dots, z_m) \in \R^{2m}$, and
$\chi$ is an arbitrary real-valued function yielding the gauge
transformation. Since vortices are exponentially localized,
for large inter-vortex separations such test functions are
approximate -- but not exact --
solutions of the stationary Ginzburg-Landau equations.
We measure the inter-vortex distance by
\[
  R(\zb) := \min_{j \not= k} |z_j-z_k|,
\]
and introduce the associated small parameter
$\e = \e(\zb) := R(\zb)^{-1/2} e^{-R(\zb)}$.


\paragraph{Dynamical problem.}

Now consider a time-dependent Ginzburg-Landau equation
with an initial condition close to the function
$\vzo = (\psi_{\zb_0,\chi_0}, A_{\zb_0,\chi_0})$ describing
several vortices glued together (if $\kappa > 1/2$, we take
$n_j = \pm 1$ since the $|n| \geq 2 -$ vortices are then
unstable by Theorem~\ref{thm:stab}), and ask the following
questions:
\begin{itemize}
\item
Does the solution at a later time $t$
describe well-localized vortices at some locations
$\zb = \zb(t)$ (and with a gauge transformation $\chi = \chi(t)$)?
\item
If so, what is the dynamical law of the vortex
centers $\zb(t)$ (and of $\chi (t)$)?
\end{itemize}


\paragraph{Vortex dynamics results.}

This section gives a brief description of the vortex dynamics
results in~\cite{GS2}.

\bigskip

\noindent
{\it Gradient flow (we take, for simplicity, $\g=1, \s=1$ in \eqref{GES}).}
Consider the gradient flow equations~\eqref{GES}
with initial data $(\psi(0),A(0))$ close (in the
energy norm) to some multi-vortex configuration $\vzo$. Then
\be
\la{eq:close}
  (\Psi(t),A(t)) = \vzt + O(\e \log^{1/4}(1/\e)),
\end{equation}
and the vortex dynamics is governed by the system
\be
\la{eq:seff}
  \g_{n_j} \dot{z}_j = -\nabla_{z_j} W(\zb) + O(\e^2 \log^{3/4}(1/\e)).
\end{equation}
Here $\g_n >0$, and $W(\zb)$ is the effective vortex interaction
energy, of order $\e$, given below. 

These statements hold for only as long as the path
$\zb(t)$ does not violate a condition of large separation,
though in the {\it repulsive case}, when
$\lam > 1/2$ and $n_j = +1$ (or $n_j = -1$) for all $j$,
the above statements hold for all time $t$.

\bigskip

\noindent
{\it Maxwell-Higgs equations.}
For the Maxwell-Higgs equations~\eqref{MHeq}
with initial data $(\Psi(0),A(0))$ close (in the energy norm)
to some $\vzo$ (and with appropriately small initial momenta),
\be
\la{eq:close2}
  \|(\Psi(t),A(t)) - \vzt\|_{H^1} +
  \|(\p_t \psi(t),\p_t A(t)) - \p_t \vzt\|_{L^2}
  = o(\sqrt{\e})
\end{equation}
with
\be
\la{eq:heff}
  \gamma_{n_j}\ddot{z}_j = -\nabla_{z_j} W(\zb(t)) + o(\e)
\end{equation}
for times up to (approximately) order
$\frac{1}{\sqrt{\e}}\log\left(\frac{1}{\e}\right)$.


\bigskip

\noindent
{\it The effective vortex interaction energy.} The interaction energy of a multi-vortex configuration is defined as
\be
\la{W}   W(\zb) = \E(\vz) - \sum_{j=1}^n \E(\Psi^{(n_j)},A^{(n_j)}).
\end{equation}
For $\kappa > 1/2$, we have for large $R(\zb)$,
$$W(\zb) \sim (const) \sum_{j \not= k} n_j n_k
  \frac{e^{-|z_j-z_k|}}{\sqrt{|z_j-z_k|}}.$$

\paragraph{Some ideas behind the proofs (Maxwell-Higgs case).}
\begin{itemize}
\item
{\it Multi-vortex manifold.}
Multi-vortex configurations $\vz$ comprise an (infinite-
dimensional, due to gauge transformations) manifold
\[
  M := \{ \vz \; | \; \zb \in \R^{2m}, \chi \in H^2 \},
\]
(together with appropriate momenta in the Maxwell-Higgs case)
made up of approximate static solutions. The interaction energy \eqref{W} of a multi-vortex configuration gives rise to a reduced Hamiltonian system on $M$,
using the restriction of the natural symplectic form to $M$
-- this is the leading-order vortex motion law.
\item
{\it Symplectic orthogonality and effective dynamics.}
The {\it effective dynamics} on $M$ is determined by
demanding that the deviation of the solution from
$M$ be symplectically orthogonal to the tangent space to
$M$. Informally,
\[
  (\Psi(t), A(t)) - \vzt \;\; \perp \;\;
  \J \; T_{\vzt} M.
\]
The tangent space is composed of infinitesimal (approximate)
symmetry transformations -- that is, independent motions of the
vortex centers, and gauge transformations.
\item
{\it Stability and coercivity.}
The manifold $M$ inherits a stability property from the stability
of its basic building blocks -- the $n$-vortex solutions --
as described in Section~\ref{sec:stab}. 
The stability property is reflected in the fact that
the linearized operator around a multi-vortex
\[
  L_{\zb,\chi} = \E''(\vz)
\]
is coercive in directions symplectically orthogonal
to the tangent space of $M$:
\[
  \xi \; \perp \; \J \; T_{\vz} M \;
  \implies \; \langle \xi, L_{\zb,\chi} \xi \rangle > 0.
\]
\item
{\it Approximately conserved Lyapunov functionals.}
Thus the quadratic form $\langle \xi, L_{\zb,\chi} \xi \rangle$, where $\xi:=(\Psi(t), A(t)) - \vzt$, controls
the deviation of the solution from the multi-vortex manifold,
and furthermore is approximately conserved -- this gives
long-time control of the deviation. Finally, approximate
conservation of the reduced energy $W(\vzt)$ is used to control
the difference between the effective dynamics, and the
leading-order vortex motion law.
\end{itemize}

\appendix

 \section{Parametrization of the equivalence classes $[\LAT]$} \label{sec:param}
In this appendix, we present some standard results about lattices. We show that  lattice shapes can be parametrized by  points $\tau$ in  the fundamental domain,  $\Pi^+/SL(2, \Z)$, of  the modular group  $SL(2, \Z)$ acting on the Poicar\'e half-plane $\Pi^+$.
%
\DETAILS{ and  use translations and  rotations, if necessary, to bringing a lattice $\LAT$, satisfying the quantization condition \eqref{quant-cond},  into the form $$\LAT=
r (\Z+\tau\Z),$$ where $r>0$ and $\tau\in \C$  is a complex number with $\Im \tau>0$. 
Now, the parameter $\tau$ determines the basis   $(r ,\  r\tau)$  of the lattice $\LAT_\tau$ and any two bases in $\LAT_\tau$  are related by a modular map  $ (\nu_1,\ \nu_2)\ra (r\nu_2',\ r\nu_1')=r(\al\nu_1+\beta,\ \g\nu_2 +\del)$, where  $\al, \beta, \g, \del\in \Z$,  and 
$\al \del-\beta \g=1$ (with the matrix  $\left( \begin{array}{cc} \al & \beta \\ \g & \del \end{array} \right)$, an element  of the modular group  $SL(2, \Z)$). In particular, replacing the basis  $r(1, \tau)$  by the basis $r( \g \tau+\del,\  \al \tau+\beta )$, results in the shape parameter $\tau$ being mapped as $ \tau\ra g\tau:=\frac{\al\tau+\beta}{\g\tau+\del}$. Hence any quantity, $f$, which depends of the lattice, and not on 
the choice a basis in the lattice, is invariant under  the modular group  $SL(2, \Z)$ and therefore is determined entirely by its values on 
 the fundamental domain,}
 %

Every lattice in $\R^2$ can be written as $\LAT=\Z\nu_1+\Z\nu_2$, where  $ (\nu_1,\ \nu_2)$ is a  basis in $\R^2$.  Given a  basis $ (\nu_1,\ \nu_2)$  in $\R^2$ and   identifying $\R^2$ with $\C$, via the map $(x_1, x_2)\ra x_1+i x_2$, we define the complex number $\tau=\nu_2/\nu_1$, called the shape parameter. 
We can choose a basis so that  $\Im\tau>0$, which we assume from now on.  
 Clearly, $\tau$ is independent of translations, rotations and dilatations of the lattice and therefore depends on its equivalence class only.

 Any two bases, $ (\nu_1,\ \nu_2)$ and $ (\nu_1',\ \nu_2')$ span the same lattice $\LAT$  iff they are related 
 as $ (\nu_1',\ \nu_2')=(\al\nu_1+\beta,\ \g\nu_2 +\del)$, where  $\al, \beta, \g, \del\in \Z$,  and 
$\al \del-\beta \g=1$ (i.e. the matrix  $\left( \begin{array}{cc} \al & \beta \\ \g & \del \end{array} \right)$, an element  of the modular group  $SL(2, \Z)$). Under this map, the shape parameter $\tau=\nu_2/\nu_1$ is being mapped into  $\tau'=\nu'_2/\nu'_1$ as $ \tau\ra \tau'=g\tau$, where $g\tau:=\frac{\al\tau+\beta}{\g\tau+\del}$. Thus, up to rotation and dilatation, the lattices are in one-to-one correspondence with points $\tau$ in the fundamental domain, $\Pi^+/SL(2, \Z)$, of  the modular group  $SL(2, \Z)$ acting on   the Poicar\'e half-plane  $\Pi^+$. 
Explicitly  (see Fig. \ref{fig:PoincareStrip}),   
\begin{align}\label{fund-domSL2Z} \{\tau\in \C: \Im\tau > 0,\ |\tau| \geq 1,\ -\frac{1}{2} < \Re\tau \leq \frac{1}{2} \}. 
\end{align} 


Furthermore, any quantity, $f$, which depends of the lattice equivalence classes, can be thought of as a function of $\tau,\ \Im\tau>0$, 
 invariant under  the modular group  $SL(2, \Z)$ and therefore is determined entirely by its values on  the fundamental domain, \eqref{fund-domSL2Z}.

 \section{Automorphy factors} \label{sec:gs}  
 
  We list some important properties of $g_s$: 
\begin{itemize}
\item   If $(\Psi, A)$ satisfies \eqref{gauge-per'} with  $g_s(x)$, then $T^{gauge}_\chi (\Psi, A)$ satisfies \eqref{gauge-per'} with  $g_s(x)\ra g'_s(x) $, where 
\begin{equation}\label{gs-equiv}
    g'_s(x) = g_s(x) + \chi (x + s) - \chi(x).
\end{equation}
 
\item 	
The functions $g_s(x)  = \frac{b}{2} s \wedge x + c_s$, where $b$ satisfies $b |\Omega|\in 2\pi\Z$ 
  and $c_s$ are numbers satisfying $ c_{s+t} - c_s - c_t - \frac{1}{2} b s \wedge t \in 2\pi\Z$, satisfies  \eqref{cocycle-cond}.

  \item By the cocycle condition \eqref{cocycle-cond},  for any basis $\{\nu_1, \nu_2\}$ in $\cL$,  the quantity
\begin{equation}\label{gs-equiv}
    c(g_s) =\frac{1}{2\pi} (g_{\nu_2}(x+\nu_1) - g_{\nu_2}(x) - g_{\nu_1}(x+\nu_2) + g_{\nu_1}(x)) 
\end{equation}
 is independent of $x$ and of the choice of the basis $\{\nu_1, \nu_2\}$ and is an integer. 
 
\item   Every exponential  $g_s$ satisfying the cocycle condition \eqref{cocycle-cond}   is equivalent to  the exponent 
\begin{equation}\label{gs-spec} \frac{b}{2} s \wedge x + c_s,\end{equation} for 
  $b$ and  $c_s$   satisfying $b |\Omega|= 2\pi c(g_s)$ and 
  \begin{equation}\label{cs-eq}  c_{s+t} - c_s - c_t - \frac{1}{2} b s \wedge t \in 2\pi\Z.\end{equation}

\item 	The condition \eqref{cocycle-cond} implies 
the magnetic flux quantization \eqref{flux-quant}: 
 \begin{equation}\label{flux-cs}\frac{1}{2\pi} \int_\Omega \Curl A =\deg \Psi  
=c(g_s). \end{equation}
\end{itemize}

Indeed, the first, second and third statements are straightforward.  

For the fourth property, see e.g. \cite{Eil, Odeh, Takac, TS2}, though in these papers it is formulated differently. In the present formulation it was shown  by A. Weil and generalized in  \cite{Gun1}. 

To prove the  fifth statement,  
we note that  by Stokes' theorem,  the magnetic flux through a lattice cell $\Om$ is   $\int_\Omega \Curl A = \int_{\p\Omega} A$, 
is given by \begin{align*}
  \int_0^1  &\big[ \nu_1 \cdot (A(a\nu_1 + \nu_2) - A(a\nu_1)) - \nu_2 \cdot (A(a\nu_2 + \nu_1) - A(a\nu_2))  \big] da \\
 &= \int_0^1 \big[ \nu_1 \cdot \nabla g_{\nu_2}(a\nu_1) - \nu_2 \cdot \nabla g_{\nu_1}(a\nu_2) \big] da, 
 \end{align*}
which, by \eqref{cocycle-cond},  gives
	$\int_\Omega \CURL A 	= g_{\nu_2}(\nu_1) - g_{\nu_2}(0) - g_{\nu_1}(\nu_2) + g_{\nu_1}(0) \in 2\pi\Z. $ 

\begin{itemize}
\item The exponentials  $g_s$ satisfying the cocycle condition \eqref{cocycle-cond} are classified by the irreducible representation of the group of lattice translations. 
\item For a family $g_s$ of functions satisfying \eqref{cocycle-cond}, there exists 
a continuous pair $(\Psi, A)$ satisfying \eqref{gauge-per'} with this family.
 \end{itemize}   
The first property follows from the fact that  $c_s$'s satisfying $c_{s+t} - c_s - c_t - \frac{1}{2} b s \wedge t \in 2\pi\Z$ 
are classified by the irreducible representation of the group of lattice translations. 

For the second property,   given a family $g_s$ of functions satisfying \eqref{cocycle-cond},   
 equivariant functions  $u=(\Psi, A)$ for $g_s$ are identified with {\it sections of the vector bundle}
  $$\R^2\times (\C\times \R^2)/\cL,$$
  with the base manifold $\R^2/\cL=\Om$ and the projection 
 $p: [(x, u)]\ra [x]$, where $[(x, u)]$ and $[x]$ are the equivalence classes of $(x, u)$ and $x$, under the action of the group $\cL$ on $\R^2\times  (\C\times \R^2)$ and on  $\R^2$, given by $$s: (x, u)\ra (x + s,  T^{\rm gauge}_{g_s(x)} u)\ \text{  and }\  s: x\ra x+s,$$
  respectively.



\medskip
\noindent {\bf Remark.}    In algebraic geometry and number theory, $e^{ig_s(x)}$ is called the automorphy factor and the factors  $e^{ig'_s(x)}$ and $e^{ig_s(x)}$ satisfying $g'_s(x) = g_s(x) + \chi (x + s) - \chi(x)$, for some $\chi (x)$, are said to be equivalent.  
A function $\Psi$ satisfying $T^{\rm trans}_s \Psi =e^{i g_s} \Psi$ is called $e^{i g_s}-$theta function. 

\end{document}